\documentclass{article}

\usepackage{amsmath}
\usepackage{amsthm}
\usepackage[T1]{fontenc}
\usepackage{amsfonts}
\usepackage{amssymb}
\usepackage{float}
\usepackage{url}
\usepackage{tikz}
\usetikzlibrary{automata}
\usetikzlibrary{shadows}
\usetikzlibrary{arrows}
\usetikzlibrary{shapes}
\usetikzlibrary{decorations.pathmorphing}
\usetikzlibrary{fit}
\usepackage{algorithm}
\usepackage{algorithmic}

\tikzstyle{every picture}=[
  >=stealth', shorten >=1pt, node distance=1.44cm,auto,bend angle=45,initial text=,
  every state/.style={inner sep=0.75mm, minimum size=1mm},font=\scriptsize,
]

\newcommand{\cqfd}{\ }

\newtheorem{definition}{Definition}
\newtheorem{lemma}{Lemma}
\newtheorem{example}{Example}
\newtheorem{proposition}{Proposition}
\newtheorem{theorem}{Theorem}

\newtheorem{corollary}{Corollary}

\begin{document} 

  \title{Two-Sided Derivatives for Regular Expressions and for Hairpin Expressions}
  
  \author{
    Jean-Marc Champarnaud 
    \and Jean-Philippe Dubernard 
    \and Hadrien Jeanne 
    \and Ludovic Mignot
  } 
  
  \maketitle%
  
   \begin{abstract}
The aim of this paper is to design the polynomial construction of a finite recognizer for hairpin completions of regular languages.
This is achieved by considering completions as new expression operators
and by applying derivation techniques to the associated extended expressions called hairpin expressions.
More precisely, we extend partial derivation of regular expressions to two-sided partial derivation of hairpin expressions
and we show how to deduce a recognizer for a hairpin expression from its two-sided derived term automaton,
providing an alternative proof of the fact that hairpin completions of regular languages are linear context-free.  

  \end{abstract}

\section{Introduction}\label{se:int}

The aim of this paper is to design the polynomial construction of a finite recognizer for hairpin completions of regular languages.
Given an integer $k>0$ and an involution $\mathrm{H}$ over an alphabet $\Gamma$,
the hairpin $k$-completion of two languages $L_1$ and $L_2$ over $\Gamma$ is the language
$\mathrm{H}_k(L_1,L_2)=\{\alpha\beta\gamma\mathrm{H}(\beta)\mathrm{H}(\alpha)\mid \alpha,\beta,\gamma\in\Gamma^*\wedge (\alpha\beta\gamma\mathrm{H}(\beta)\in L_1 \vee \beta\gamma\mathrm{H}(\beta)\mathrm{H}(\alpha)\in L_2) \wedge |\beta|=k\}$ (see Figure~\ref{HP}).
Hairpin completion has been deeply studied~\cite{BLMM06,CM01,CMM06,DKM09,DKM12,ILMM11,KKS12,KSK11,Kop11,MMM09,MMM10,MM07,MMY09,MMY10,MMM06}.
The hairpin completion of formal languages has been introduced in~\cite{CMM06} by reason of its application to biochemistry.
It aroused numerous studies that investigate theoretical and algorithmic properties of hairpin completions or related operations (see for example~\cite{KSK11,MMM09,MMY09}).
One of the most recent result concerns the problem of deciding regularity of hairpin completions of regular languages;
it can be found in~\cite{DKM12} as well as a complete bibliography about hairpin completion. 

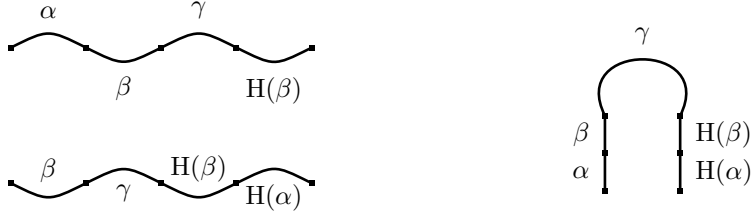
\begin{figure}[H]
\label{HP}
\begin{minipage}{0.45\linewidth}
  \centerline{
    \begin{tikzpicture}[font=\normalsize]
      \node [fill=black,inner sep=1pt]  at (0,0) {};
      \node [fill=black,inner sep=1pt]  at (1,0) {};
      \node [fill=black,inner sep=1pt]  at (2,0) {};
      \node [fill=black,inner sep=1pt]  at (3,0) {};
      \node [fill=black,inner sep=1pt]  at (4,0) {};
      \draw[line width=1pt] (0,0) .. controls (0.5,0.25) ..  node[above=2pt] {$\alpha$} (1,0) .. controls (1.5,-0.25) .. node[below=2pt] {$\beta$} (2,0)  .. controls (2.5,0.25) .. node[above=2pt] {$\gamma$} (3,0) .. controls (3.5,-0.25) .. node[below=2pt] {$\mathrm{H}(\beta)$}(4,0);
    \end{tikzpicture}
  }
  
  \vspace{\baselineskip}
  
  \centerline{
    \begin{tikzpicture}[font=\normalsize]
      \node [fill=black,inner sep=1pt]  at (0,0) {};
      \node [fill=black,inner sep=1pt]  at (1,0) {};
      \node [fill=black,inner sep=1pt]  at (2,0) {};
      \node [fill=black,inner sep=1pt]  at (3,0) {};
      \node [fill=black,inner sep=1pt]  at (4,0) {};
      \draw[line width=1pt] (0,0) .. controls (0.5,-0.25) ..  node[above=2pt] {$\beta$} (1,0) .. controls (1.5,0.25) .. node[below=2pt] {$\gamma$} (2,0)  .. controls (2.5,-0.25) .. node[above=2pt] {$\mathrm{H}(\beta)$} (3,0) .. controls (3.5,0.25) .. node[below=2pt] {$\mathrm{H}(\alpha)$}(4,0);
    \end{tikzpicture}
  }
\end{minipage}
\hfill
\begin{minipage}{0.45\linewidth}
  \centerline{
    \begin{tikzpicture}[font=\normalsize]
      \node [fill=black,inner sep=1pt]  at (-1,-1) {};
      \node [fill=black,inner sep=1pt]  at (-1,-0.5) {};
      \node [fill=black,inner sep=1pt]  at (-1,0) {};
      \node [fill=black,inner sep=1pt]  at (0,0) {};
      \node [fill=black,inner sep=1pt]  at (0,-0.5) {};
      \node [fill=black,inner sep=1pt]  at (0,-1) {};
      \draw[line width=1pt] (-1,-1)-- node[left =2pt] {$\alpha$} (-1,-0.5)--  node[left =2pt] {$\beta$}(-1,0) .. controls (-1.5,1) and (0.5,1) .. node[above =2pt] {$\gamma$} (0,0)-- node[right =2pt] {$\mathrm{H}(\beta)$} (0,-0.5)-- node[right =2pt] {$\mathrm{H}(\alpha)$} (0,-1);
    \end{tikzpicture}
  }
\end{minipage}
\caption{The Hairpin Completion.}
\end{figure}

Hairpin completions of regular languages are proved to be linear context-free from~\cite{CMM06}.
An alternative proof is presented in this paper, with a somehow more constructive approach, since it provides a recognizer for the hairpin completion. 
This is achieved by considering completions as new expression operators
and by applying derivation techniques to the associated extended expressions, that we call hairpin expressions.
Notice that a similar derivation-based approach has been used to study approximate regular expressions~\cite{CJM12}, through the definition of new distance operators.

Two-sided derivation is shown to be particularly suitable for the study of hairpin expressions.
More precisely, we extend partial derivation of regular expressions~\cite{Ant96} to two-sided partial derivation of regular expressions first
and then of hairpin expressions.
We prove that the set of two-sided derived terms of a hairpin expression $E$ over an alphabet $\Gamma$ is finite.
Hence the two-sided derived term automaton $A$ is a finite one.
Furthermore the automaton $A$ is over the alphabet $(\Gamma\cup\{\varepsilon\})^2$ and, as we prove it, the language over $\Gamma$ of such an automaton is linear context-free and not necessarily regular.
Finally we show that the language of the hairpin expression $E$ and the language over $\Gamma$ of the automaton $A$ are equal.

This paper is an extended version of the conference paper~\cite{CDJM13tmp}.
It is organized as follows. Next section gathers useful definitions and properties concerning automata and regular expressions. The notion of two-sided residual of a language is introduced in Section~\ref{sec:2sided}, as well as the related notion of $\Gamma$-couple automaton.
In Section~\ref{sec:hairpinComp}, hairpin completions of regular languages and their two-sided residuals are investigated.
The two-sided partial derivation of hairpin expressions is considered in Section~\ref{2sideddder}, leading to the construction of a finite recognizer.
A specific case is examined in Section~\ref{sec:H0}.

\section{Preliminaries}

  An \emph{alphabet} is a finite set of distinct symbols. Given an alphabet $\Sigma$, we denote by $\Sigma^*$ the set of all the words over $\Sigma$. The empty word is denoted by $\varepsilon$. A \emph{language} over $\Sigma$ is a subset of $\Sigma^*$.  
  The three operations $\cup$, $\cdot$ and $^*$ are defined for any two languages $L_1$ and $L_2$ over $\Sigma$ by: $L_1\cup L_2=\{w\in\Sigma^*\mid w\in L_1\ \vee\ w\in L_2\}$, $L_1\cdot L_2=\{w_1w_2\in\Sigma^*\mid w_1\in L_1\ \wedge\ w_2\in L_2\}$, $L_1^*=\{\varepsilon\}\cup \{w_1\cdots w_k\in\Sigma^*\mid \forall j\in\{1,\ldots,k\},w_j\in L_1\}$.  
  The family of \emph{regular languages} over $\Sigma$ is the smallest family $\mathcal{F}$ closed under the three operations $\cup$, $\cdot$ and $^*$ satisfying $\emptyset\in\mathcal{F}$ and $\forall a\in\Sigma,\ \{a\}\in\mathcal{F}$.   
  Regular languages can be represented by \emph{regular expressions}. A \emph{regular expression} over $\Sigma$ is inductively defined by: $E=a$, $E=\varepsilon$, $E=\emptyset$, $E=F+G$, $E=F\cdot G$, $E=F^*$, where $a$ is any symbol in $\Sigma$ and $F$ and $G$ are any two regular expressions over $\Sigma$. The \emph{width of} $E$ is the number of occurrences of symbols in $E$, and its \emph{star number} the number of occurrences of the operator $^*$.  
  The language \emph{denoted by} $E$ is the language $L(E)$ inductively defined by: $L(A)=\{a\}$, $L(\varepsilon)=\{\varepsilon\}$, $L(\emptyset)=\emptyset$, $L(F+G)=L(F)\cup L(G)$, $L(F\cdot G)=L(F)\cdot L(G)$, $L(F^*)=(L(F))^*$, where $a$ is any symbol in $\Sigma$ and $F$ and $G$ are any two regular expressions over $\Sigma$. The language denoted by a regular expression is regular.
  
  Let $w$ be a word in $\Sigma^*$ and $L$ be a language. The \emph{left residual} (resp. \emph{right residual}) of $L$ w.r.t. $w$ is the language $w^{-1}(L)=\{w'\in\Sigma^*\mid ww'\in L\}$ (resp. $(L)w^{-1}=\{w'\in\Sigma^*\mid w'w\in L\}$). It has been shown
   that the set of the left residuals (resp. right residuals) of a language is a finite set if and only if the language is regular.

An \emph{automaton} (or a \emph{NFA}) over an alphabet $\Sigma$ is a $5$-tuple $A=(\Sigma,Q,I,F,\delta)$ where $\Sigma$ is an alphabet, $Q$ a finite set of \emph{states}, $I\subset Q$ the set of \emph{initial states}, $F\subset Q$ the set of \emph{final states} and $\delta$ the \emph{transition function} from $Q\times \Sigma$ to $2^Q$. The domain of the function $\delta$ can be extended to $2^Q\times \Sigma^*$ as follows: for any word $w$ in $\Sigma^*$, for any symbol $a$ in $\Sigma$, for any set of states $P\subset Q$, for any state $p\in Q$,  $\delta(P,\varepsilon)=P$, $\delta(p,aw)=\delta(\delta(p,a),w)$ and $\delta(P,w)=\bigcup_{p\in P}\delta(p,w)$.

The \emph{language recognized} by the automaton $A$ is the set $L(A)=\{w\in\Sigma^*\mid \delta(I,w)\cap F\neq\emptyset\}$. Given a state $q$ in $Q$, the \emph{right language} of $q$ is the set $\overrightarrow{L}(q)=\{w\in\Sigma^*\mid \delta(q,w)\cap F\neq\emptyset\}$. It can be shown that \textbf{(1)} $L(A)=\bigcup_{i\in I} \overrightarrow{L}(i)$, \textbf{(2)} $\overrightarrow{L}(q)=\{\varepsilon\mid q\in F\} \cup (\bigcup_{a\in\Sigma,p\in\delta(q,a)}\{a\}\cdot \overrightarrow{L}(p))$ and \textbf{(3)} $a^{-1}(\overrightarrow{L}(q))=\bigcup_{p\in\delta(q,a)}  \overrightarrow{L}(p)$.

Kleene Theorem~\cite{Kle56} asserts that a language is regular if and only if there exists an NFA that recognizes it. As a consequence, for any language $L$, there exists a regular expression $E$ such that $L(E)=L$ if and only if there exists an NFA $A$ such that $L(A)=L$. Conversion methods from an NFA to a regular expression
and \emph{vice versa}
 have been deeply studied. In this paper, we focus on the notion of partial derivative defined by Antimirov~\cite{Ant96}\footnote{Partial derivation is investigated in the more general framework of weighted expressions in~\cite{LS01}.}.

Given a regular expression $E$ over an alphabet $\Sigma$ and a word $w$ in $\Sigma^*$, the \emph{left partial derivative} of $E$ w.r.t. $w$ is the set $\frac{\partial}{\partial_w}(E)$ of regular expressions satisfying: $\bigcup_{E'\in\frac{\partial}{\partial_w}(E)} L(E')=w^{-1}(L(E))$.

This set is inductively computed as follows: for any two regular expressions $F$ and $G$, for any word $w$ in $\Sigma^*$ and for any two distinct symbols $a$ and $b$ in $\Sigma$,

\centerline{
  $\frac{\partial}{\partial_a}(a)=\{\varepsilon\}$, $\frac{\partial}{\partial_a}(b)=\frac{\partial}{\partial_a}(\varepsilon)=\frac{\partial}{\partial_a}(\emptyset)=\emptyset$,}
  
  \centerline{
    $\frac{\partial}{\partial_a}(F+G)=\frac{\partial}{\partial_a}(F)\cup \frac{\partial}{\partial_a}(G)$, $\frac{\partial}{\partial_a}(F^*)=\frac{\partial}{\partial_a}(F)\cdot F^*$,}
    
    \centerline{
      $\frac{\partial}{\partial_a}(F\cdot G)=
        \left\{
          \begin{array}{l@{\ }l}
            \frac{\partial}{\partial_a}(F)\cdot G \cup \frac{\partial}{\partial_a}(G) & \text{ if }\varepsilon\in L(F),\\
            \frac{\partial}{\partial_a}(F)\cdot G & \text{ otherwise,}\\
          \end{array}
        \right.$
    } 
    
    \centerline{
      $\frac{\partial}{\partial_{aw}}(F)=\frac{\partial}{\partial_w}(\frac{\partial}{\partial_a}(F))$, $\frac{\partial}{\partial_{\varepsilon}}(F)=\{F\}$,
    }
    
    where for any set $\mathcal{E}$ of regular expressions, for any word $w$ in $\Sigma^*$, for any regular expression $F$,  $\frac{\partial}{\partial_w}(\mathcal{E})=\bigcup_{E\in\mathcal{E}}\frac{\partial}{\partial_w}(E)$ and $\mathcal{E}\cdot F=\bigcup_{E\in\mathcal{E}} \{E\cdot F\}$. Any expression appearing in a left partial derivative is called a \emph{left derived term}.    
    Similarly, the \emph{right partial derivative} of a regular expression $E$ over an alphabet $\Sigma$ w.r.t. a word $w$ in $\Sigma^*$ is the set $(E)\frac{\partial}{\partial_w}$ inductively defined as follows for any two regular expressions $F$ and $G$, for any word $w$ in $\Sigma^*$ and for any two distinct symbols $a$ and $b$ in $\Sigma$,

\centerline{
  $(a)\frac{\partial}{\partial_a}=\{\varepsilon\}$, $(b)\frac{\partial}{\partial_a}=(\varepsilon)\frac{\partial}{\partial_a}=(\emptyset)\frac{\partial}{\partial_a}=\emptyset$,}
  
  \centerline{
    $(F+G)\frac{\partial}{\partial_a}=(F)\frac{\partial}{\partial_a}\cup (G)\frac{\partial}{\partial_a}$, $(F^*)\frac{\partial}{\partial_a}=F^* \cdot  (F)\frac{\partial}{\partial_a}$,}
    
    \centerline{
      $(F\cdot G)\frac{\partial}{\partial_a}=
        \left\{
          \begin{array}{l@{\ }l}
            F\cdot (G)\frac{\partial}{\partial_a} \cup (F)\frac{\partial}{\partial_a} & \text{ if }\varepsilon\in L(G),\\
            F\cdot (G)\frac{\partial}{\partial_a} & \text{ otherwise,}\\
          \end{array}
        \right.$
    } 
    
    \centerline{
      $(F)\frac{\partial}{\partial_{aw}}=((F)\frac{\partial}{\partial_a})\frac{\partial}{\partial_w}$, $(F)\frac{\partial}{\partial_{\varepsilon}}=\{F\}$,
    }
    
    where for any set $\mathcal{E}$ of regular expressions, for any word $w$ in $\Sigma^*$, for any regular expression $F$,  $(\mathcal{E})\frac{\partial}{\partial_w}=\bigcup_{E\in\mathcal{E}}(E)\frac{\partial}{\partial_w}$ and $F\cdot \mathcal{E}=\bigcup_{E\in\mathcal{E}}\{F\cdot E\}$. Any expression appearing in a right partial derivative is called a \emph{right derived term}.    
We denote by $\overleftarrow{\mathcal{D}_E}$ (resp. $\overrightarrow{\mathcal{D}_E}$) the  set of left (resp. right) derived terms of the expression $E$.    
    From the set of left derived terms of a regular expression $E$ of width $n$, Antimirov defined in~\cite{Ant96} the \emph{derived term automaton} $A$ of $E$ and showed that $A$ is a $k$-state NFA that recognizes $L(E)$, with $k\leq n+1$.
      
    A language over an alphabet $\Gamma$ is said to be \emph{linear context-free} if it can be generated by a linear grammar, that is a grammar equipped with productions in one of the following forms:
    \begin{enumerate}
      \item $A\rightarrow x B y$, where $A$ and $B$ are any two non-terminal symbols, and $x$ and $y$ are any two symbols in $\Gamma\cup\{\varepsilon\}$ such that $(x,y)\neq(\varepsilon,\varepsilon)$,
      \item $A \rightarrow \varepsilon$, where $A$ is any non-terminal symbol.
    \end{enumerate}
     
    Notice that the family of regular languages is strictly included into the family of linear context-free languages. In the following, we will consider combinations of left and right partial derivatives in order to deal with non-regular languages.
    
\section{Two-sided Residuals of a Language and Couple NFA}\label{sec:2sided}

  In this section, we extend residuals to two-sided residuals. This operation is the composition of left and right residuals, but it is more powerful than classical residuals since it allows to compute a finite subset of the set of residuals even for non-regular languages, which leads to the construction of a derivative-based finite recognizer.

\begin{definition}\label{def two side quot}
  Let $L$ be a language over an alphabet $\Gamma$ and let $u$ and $v$ be two words in $\Gamma^*$. The \emph{two-sided residual} of $L$ w.r.t. $(u,v)$ is the language $(u,v)^{-1}(L)=\{w\in\Gamma^*\mid uwv\in L\}$.
\end{definition}

As above-mentioned, the two-sided residual operation is the composition of the two operations of left and right residuals.

\begin{lemma}\label{lem 2sidequot eq quot 2side}
  Let $L$ be a language over an alphabet $\Gamma$ and $u$ and $v$ be two words in $\Gamma^*$. Then: $(u,v)^{-1}(L)=(u^{-1}(L))v^{-1}=u^{-1}((L)v^{-1})$.
\end{lemma}
\begin{proof}
  Let $w$ be a word in $\Gamma^*$.
  
  \centerline{
    \begin{tabular}{l@{\ }l@{\ }l@{\ }l} 
      $w\in (u^{-1}(L))v^{-1}$ &  $\Leftrightarrow$ $wv \in u^{-1}(L)$ & $\Leftrightarrow$ $uwv \in L$ & $\Leftrightarrow$ $(u,v)^{-1}(L)$\\
       & $\Leftrightarrow$ $uwv \in L$ & $\Leftrightarrow$ $uw \in (L)v^{-1}$ & $\Leftrightarrow$ $w \in u^{-1}((L)v^{-1})$.\\
    \end{tabular}
  }
  \cqfd
\end{proof}

\begin{corollary}
  Let $L$ be a language over an alphabet $\Gamma$ and $u$ and $v$ be two words in $\Gamma^*$. Then: $\varepsilon\in (u,v)^{-1}(L)$ $\Leftrightarrow$ $uv\in L$.
\end{corollary}

It is a folk knowledge that NFAs are related to left residual computation according to the following assertion \textbf{(A)}: in an NFA $(\Sigma,Q,I,F,\delta)$, a word $aw$ belongs to $\overrightarrow{L}(q)$ with $q\in Q$ if and only if $w$ belongs to $a^{-1}(\overrightarrow{L}(q))=\bigcup_{q'\in\delta(q,a)} \overrightarrow{L}(q')$.
Since a two-sided residual w.r.t. a couple $(x,y)$ of symbols in an alphabet $\Gamma$ is by definition the combination of a left residual w.r.t. $x$ and of a right residual w.r.t. $y$, the assertion \textbf{(A)} can be extended to two-sided residuals by introducing \emph{couple NFAs} equipped with transitions labelled by couples of symbols in $\Gamma$. The notion of right language of a state is extended to the one of $\Gamma$-right language as follows:
if a given word $w$ in $\Gamma^*$ belongs to the $\Gamma$-right language of a state $q'$ and if there exists a transition from a state $q$ to $q'$ labelled by a couple $(x,y)$, then the word $xwy$ belongs to the $\Gamma$-right language of $q$.

More precisely, given an alphabet $\Gamma$, we set $\Sigma_{\Gamma}=\{(x,y)\mid x,y\in \Gamma\cup\{\varepsilon\} \wedge\ (x,y)\neq (\varepsilon,\varepsilon)\}$. We consider the mapping $\mathrm{Im}$ from $(\Sigma_\Gamma)^*$ to $\Gamma^*$ inductively defined for any word $w$ in $(\Sigma_\Gamma)^*$ and for any symbol $(x,y) \in \Sigma_\Gamma$ by: $\mathrm{Im}(\varepsilon)=\varepsilon$ and $\mathrm{Im}((x,y)\cdot w)=x\cdot \mathrm{Im}(w) \cdot y$.
Notice that this mapping was introduced by Sempere~\cite{Semp00} in order to compute the language denoted by a \emph{linear expression}. Linear expressions denote linear context-free languages, and are equivalent to the \emph{regular-like expressions} of Brzozowski~\cite{Brz68}.
  
\begin{definition}\label{def lang couple nfa}
  Let $A=(\Sigma,Q,I,F,\delta)$ be an NFA. The NFA $A$ is a \emph{couple NFA} if there exists an alphabet $\Gamma$ such that $\Sigma\subset \Sigma_\Gamma$. In this case, $A$ is called a $\Gamma$-\emph{couple NFA}.
  The $\Gamma$\emph{-language} of a $\Gamma$-couple NFA $A$ is the subset $L_\Gamma(A)$ of $\Gamma^*$ defined by: $L_\Gamma(A)=\{\mathrm{Im}(w)\mid w\in L(A)\}$.
\end{definition}

The definition of right languages and their classical properties extend to couple NFAs as follows. Let $A=(\Sigma,Q,I,F,\delta)$ be  a $\Gamma$-couple NFA and $q$ be a state in $Q$. The $\Gamma$\emph{-right language} of $q$ is the subset $\overrightarrow{L}_\Gamma(q)$ of $\Gamma^*$ defined by: $\overrightarrow{L}_\Gamma(q)=\{\mathrm{Im}(w)\mid w\in \overrightarrow{L}(q)\}$.

\begin{lemma}\label{lem prop triv couple nfa}
  Let $A=(\Sigma,Q,I,F,\delta)$ be  a $\Gamma$-couple NFA and $q$ be a state in $Q$. Then: $L_\Gamma(A)=\bigcup_{i\in I} \overrightarrow{L}_\Gamma(i)$.  
\end{lemma}
\begin{proof}
  Trivially deduced from Definition~\ref{def lang couple nfa}, from definition of $\Gamma$-right languages and from the fact that $L(A)=\bigcup_{i\in I} \overrightarrow{L}(i)$.
  \cqfd
\end{proof}

\begin{lemma}\label{lem calcul lang droit gamma}
  Let $A=(\Sigma,Q,I,F,\delta)$ be  a $\Gamma$-couple NFA and $q$ be a state in $Q$. Then: $\overrightarrow{L}_\Gamma(q)=\{\varepsilon\mid q\in F\}\cup \bigcup_{(x,y)\in\Sigma, q'\in \delta(q,(x,y))} \{x\}\cdot \overrightarrow{L}_\Gamma(q')\cdot \{y\}$.
\end{lemma}
\begin{proof}
  Trivially deduced from Definition~\ref{def lang couple nfa}, from definition of $\Gamma$-right languages and from the fact that $\overrightarrow{L}(q)=\{\varepsilon\mid q\in F\}\cup \bigcup_{a\in\Sigma, q'\in \delta(q,a)} \{a\}\cdot \overrightarrow{L}(q')$.
  \cqfd
\end{proof}

\begin{corollary}\label{cor prop triv couple nfa}
  Let $A=(\Sigma,Q,I,F,\delta)$ be  a $\Gamma$-couple NFA, $(x,y)$ be a couple in $\Sigma_\Gamma$ and $q$ be a state in $Q$. Then: $(x,y)^{-1}(\overrightarrow{L}_\Gamma(q))=\bigcup_{q'\in \delta(q,(x,y))} \overrightarrow{L}_\Gamma(q')$.
\end{corollary}

The following example illustrates the fact that there exist non-regular languages that can be recognized by couple NFAs.

\begin{example}
  Let $\Gamma=\{a,b\}$ and $A$ be the automaton of the Figure~\ref{fig ex couple nfa pas rat}. The $\Gamma$-language of $A$ is $L_\Gamma(A)=\{a^nb^n\mid n\in\mathbb{N}\}$.
\end{example}

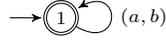
\begin{figure}[H]
  \centerline{ 
    \begin{tikzpicture}[node distance=2cm,bend angle=30]   
    \node[initial,accepting,state] (q) {$1$};
    \path[->]
      (q)   edge [swap,in=30,out=-30,loop] node {$(a,b)$} ();	            
      \end{tikzpicture}
  }  
  \caption{The Couple Automaton $A$.}
  \label{fig ex couple nfa pas rat}
\end{figure}

As a consequence there exist non-regular languages that are recognized by a couple NFA. In fact, the family of languages recognized by couple NFAs is exactly the family of linear context-free languages.

\begin{proposition}\label{prop gamma lang linear}
  The $\Gamma$-language recognized by a $\Gamma$-couple NFA is linear context-free.
\end{proposition}
\begin{proof}
  Let $A=(\Sigma,Q,I,F,\delta)$. Let us define the grammar $G=(X,V,P,S)$ by:
  \begin{itemize}
    \item $X=\Gamma$, the set of terminal symbols,
    \item $V=\{A_q\mid q\in Q\}\cup \{S\}$, the set of non-terminal symbols,
    \item $P=\{S\rightarrow A_q\mid q\in I\} \cup \{A_q\rightarrow \varepsilon\mid q\in F\}\cup \{A_q\rightarrow \alpha A_{q'} \beta \mid q'\in\delta(q,(\alpha,\beta))\}$, the set of productions,
    \item $S$, the axiom.
  \end{itemize}
  
  \begin{enumerate}
    \item Let $w$ be word in $\Gamma^*$. Let us first show that $w$ belongs to the language generated by the grammar $G_q=(X,V,P,A_q)$ if and only if it is in $\overrightarrow{L}_\Gamma(q)$, by recurrence over the length of $w$.
    \begin{enumerate}
      \item Let us suppose that $w=\varepsilon$. By construction of $G_q$, $A_q\rightarrow \varepsilon$ if and only if $q\in F$, \emph{i.e.} $\varepsilon\in \overrightarrow{L}_\Gamma(q)$.
      \item Let us suppose that $w=\alpha w' \beta$ with $(\alpha,\beta)\neq (\varepsilon,\varepsilon)$. By definition of $L(G_q)$, $w\in L(G_q)$ if there exists a symbol $A_{q'}$ in $V$ such that $A_q  \rightarrow \alpha A_{q'} \beta$ and $w'\in L(G_{q'})$. By recurrence hypothesis, it holds that $w'\in L(G_{q'}) \Leftrightarrow w'\in \overrightarrow{L}_\Gamma(q')$. Since by construction $A_q  \rightarrow \alpha A_{q'} \beta \Leftrightarrow q'\in \delta(q,(\alpha,\beta))$ and since according to Lemma~\ref{lem calcul lang droit gamma}, $\overrightarrow{L}_\Gamma(q)=\{\varepsilon\mid q\in F\}\cup \bigcup_{(x,y)\in\Sigma, q'\in \delta(q,(x,y))} \{x\}\cdot \overrightarrow{L}_\Gamma(q')\cdot \{y\}$, it holds that $w\in L(G_q) \Leftrightarrow w\in \overrightarrow{L}_\Gamma(q)$.
    \end{enumerate}
    \item Since $L(G)=\bigcup_{q\mid S\rightarrow A_q} L(G_q)$, it holds from \textbf{(1)} that $L(G)=\bigcup_{q\in I} \overrightarrow{L}_\Gamma(q)$, that equals according to Lemma~\ref{lem prop triv couple nfa} to $L(A)$.
  \end{enumerate}
  
  Finally, since the $\Gamma$-language of $A$ is generated by a linear grammar, it is linear context free.
  \cqfd
\end{proof}

\begin{proposition}\label{prop lang linear reco par couplenfa}
  The language generated by a linear grammar is recognized by a couple NFA.
\end{proposition}
\begin{proof}
  Let $G=(X,V,P,S)$ be a linear grammar. Let us define the automaton $A=(\Sigma,Q,I,F,\delta)$ by:
  \begin{itemize}
    \item $\Sigma=\Sigma_X$,
    \item $Q=V$,
    \item $I=\{S\}$,
    \item $F=\{B\in V\mid (B\rightarrow \varepsilon)\in P\}$,
    \item $B'\in \delta(B,(x,y)) \Leftrightarrow (B\rightarrow xBy) \in P$.
  \end{itemize}
  
  For any symbol $B$ in $V$, let us set $G_B=(X,V,P,B)$. Let $w$ be a word in $X^*$. Let us show by recurrence over the length of $w$ that $w\in  L(G_B) \Leftrightarrow w\in \overrightarrow{L}_X(B)$.
  \begin{enumerate}
    \item Let $w=\varepsilon$. Then $\varepsilon\in L(G_B)$ if and only if $(G_B\rightarrow \varepsilon)\in P$. By construction, it is equivalent to $B\in F$ and to $\varepsilon\in \overrightarrow{L}_X(B)$.
    \item Let us suppose that $w$ is different from $\varepsilon$. Then by recurrence hypothesis and according to Lemma~\ref{lem calcul lang droit gamma}:
    
    \centerline{
      \begin{tabular}{l@{\ }l}
        & $w\in L(G_B)$ \\
        & $\Leftrightarrow$ $\exists (x,y)\in \Sigma_X,w'\in X^*,B'\in V \mid w=xw'y \wedge (B\rightarrow xB'y)\in P \wedge w'\in L(G_{B'})$\\
        & $\Leftrightarrow$ $\exists (x,y)\in \Sigma_X,w'\in X^*,B'\in V \mid w=xw'y \wedge B'\in\delta(B,(x,y)) \wedge w'\in \overrightarrow{L}_X(B')$\\
        & $\Leftrightarrow$ $w\in \overrightarrow{L}_X(B)$\\
      \end{tabular} 
    }
  \end{enumerate}  
  
  Finally, since $L(G)=L(G_S)=\overrightarrow{L}_S(B)$, it holds from Lemma~\ref{lem prop triv couple nfa} that $L(G)=L(A)$.
  \cqfd
\end{proof}

\begin{theorem}\label{thm context free equ recocouple}
  A language is linear context-free if and only if it is recognized by a couple NFA.
\end{theorem}
\begin{proof}
  Directly from Proposition~\ref{prop gamma lang linear} and from Proposition~\ref{prop lang linear reco par couplenfa}.
  \cqfd
\end{proof}

We present here two algorithms in order to solve the membership problem\footnote{Given a language $L$ and a word $w$, does $w$ belong to $L$?} \emph{via} a couple NFA.
The Algorithm~\ref{algo MBTEST} checks whether the word $w\in\Gamma^*$ is recognized by the $\Gamma$-couple NFA $A$. It returns TRUE if there exists an initial state such that its $\Gamma$-right language contains $w$. The Algorithm~\ref{algo RightLang} checks whether the word $w\in\Gamma^*$  is in the $\Gamma$-right language of the state $q$.

\begin{algorithm}[H]
  \caption{\textrm{IsInRightLanguage}($A$,$w$,$q$)}
  \label{algo RightLang}
  \begin{algorithmic}[1]
    \REQUIRE $A=(\Sigma,Q,I,F,\delta)$ a $\Gamma$-couple NFA, $w$ a word in $\Gamma^*$, $q$ a state in $Q$
    \ENSURE Returns $(w\in \overrightarrow{L}_\Gamma(q))$ 
    \IF{$w=\varepsilon$}   
      \STATE $P$ $\leftarrow$ $(q\in F)$
    \ELSE
      \STATE $P$ $\leftarrow$ FALSE
      \FORALL{$(q,(\alpha,\beta),q')\in \delta\mid w=\alpha w' \beta$}
        \STATE $P$ $\leftarrow$ $P\ \vee\ $ \textrm{IsInRightLanguage}($A$, $w'$, $q'$)
      \ENDFOR
    \ENDIF
    \RETURN $P$
  \end{algorithmic}
\end{algorithm}

\begin{algorithm}[H]
  \caption{\textrm{MembershipTest}($A$,$w$)}
  \label{algo MBTEST}
  \begin{algorithmic}[1]
    \REQUIRE $A=(\Sigma,Q,I,F,\delta)$ a $\Gamma$-couple NFA, $w$ a word in $\Gamma^*$
    \ENSURE Returns $(w\in L_\Gamma(A))$ 
    \STATE $R$ $\leftarrow$ FALSE
    \FORALL{$i\in I$}
      \STATE $R$ $\leftarrow$ $R\ \vee\ $ \textrm{IsInRightLanguage}($A$, $w$, $i$)
    \ENDFOR
    \RETURN $R$
  \end{algorithmic}
\end{algorithm}

\begin{proposition}\label{prop algos ok}
  Let $A=(\Sigma,Q,I,F,\delta)$ be a $\Gamma$-couple NFA, $q$ be a state in $Q$ and $w$ be a word in $\Gamma^*$. The two following propositions are satisfied:
  \begin{enumerate}
    \item Algorithm~\ref{algo RightLang}: \textrm{IsInRightLanguage}($A$, $w$, $q$) returns $(w\in \overrightarrow{L}_\Gamma(q))$,
    \item Algorithm~\ref{algo MBTEST}: \textrm{MembershipTest}($A$,$w$) returns $(w\in L_\Gamma(A))$. 
  \end{enumerate}
\end{proposition}\begin{proof} 
  Let $w$ be a word in $\Gamma^*$.
  \begin{enumerate}
    \item Let us show by recurrence over the length of $w$ that the algorithm \textrm{IsInRightLanguage}($A$, $w$, $q$) returns $(w\in \overrightarrow{L}_\Gamma(q))$.  
  
  If $w=\varepsilon$, $P=TRUE$ $\Leftrightarrow$ $q\in F$ $\Leftrightarrow$ $\varepsilon  \in  \overrightarrow{L}_\Gamma(q)$.
  
  Let us suppose now that $|w|\geq 1$. Then $P=\bigvee_{(q,(\alpha,\beta),q')\in \delta\mid w=\alpha w' \beta}$ \textrm{IsInRightLanguage}($A$, $w'$, $q'$). If there is no transition $(q,(\alpha,\beta),q')\in \delta$, then trivially $w\notin \overrightarrow{L}_\Gamma(q)$. For any $(q,(\alpha,\beta),q')\in \delta$, let us notice that $(\alpha,\beta)\in\Sigma_\Gamma$. As a consequence, the length of any word $w'$ satisfying $w=\alpha w' \beta\ \wedge\ (q,(\alpha,\beta),q')\in \delta$ is strictly smaller than $|w|$. Let $w'$ be a word satisfying $w=\alpha w' \beta\ \wedge\ (q,(\alpha,\beta),q')\in \delta$. According to recurrence hypothesis, \textrm{IsInRightLanguage}($A$, $w'$, $q'$) returns $(w'\in \overrightarrow{L}_\Gamma(q'))$. Hence $P=\bigvee_{(q,(\alpha,\beta),q')\in \delta\mid w=\alpha w' \beta}$ ($w'\in \overrightarrow{L}_\Gamma(q')$). Finally, according to Lemma~\ref{lem calcul lang droit gamma}, $P=(w\in \overrightarrow{L}_\Gamma(q))$.
  
    \item Since $R=\bigvee_{i\in I}$ \textrm{IsInRightLanguage}($A$, $w$, $i$), it holds as a direct consequence that $R=\bigvee_{i\in I}(w\in \overrightarrow{L}_\Gamma(i))$. Hence, according to Lemma~\ref{lem prop triv couple nfa}, it holds that $R=(w\in L_\Gamma(A))$.
  \end{enumerate}
  
  \cqfd
\end{proof}

The following sections are devoted to hairpin completions and their two-sided residuals.  It turns out that hairpin completions are linear context-free languages. Hence, we show how to compute a couple NFA that recognizes a given hairpin completion.

\section{Hairpin Completion of a Language and its Residuals}\label{sec:hairpinComp}

Let $\Gamma$ be an alphabet. An \emph{involution} $\mathrm{f}$ over $\Gamma$ is a mapping from $\Gamma$ to $\Gamma$ satisfying for any symbol $a$ in $\Gamma$, $\mathrm{f}(\mathrm{f}(a))=a$. An \emph{anti-morphism} $\mu$ over $\Gamma^*$ is a mapping from $\Gamma^*$ to $\Gamma^*$ satisfying for any two words $u$ and $v$ in $\Gamma^*$ $\mu(u\cdot v)=\mu(v)\cdot \mu(u)$. Any mapping $\mathrm{g}$ from $\Gamma$ to $\Gamma$ can be extended as an anti-morphism over $\Gamma^*$ as follows: $\forall a\in\Gamma,$ $\forall w\in\Gamma^*$, $\mathrm{g}(\varepsilon)=\varepsilon$, $\mathrm{g}(a\cdot w)=\mathrm{g}(w)\cdot \mathrm{g}(a)$.

\begin{definition}
  Let $\Gamma$ be an alphabet and $\mathrm{H}$ be an anti-morphism over $\Gamma^*$. Let $L_1$ and $L_2$ be two languages over $\Gamma$. Let $k>0$ be an integer. The $(H,k)$-\emph{completion} of $L_1$ and $L_2$ is the language $\mathrm{H}_k(L_1,L_2)$ defined by:
  
  \centerline{
    $\mathrm{H}_k(L_1,L_2)$
  }
  
  \centerline{$=$}
  
  \centerline{$\{\alpha\beta\gamma\mathrm{H}(\beta)\mathrm{H}(\alpha)\mid \alpha,\beta,\gamma\in\Gamma^*\wedge (\alpha\beta\gamma\mathrm{H}(\beta)\in L_1 \vee \beta\gamma\mathrm{H}(\beta)\mathrm{H}(\alpha)\in L_2) \wedge |\beta|=k\}$.
  }
\end{definition}

The $(\mathrm{H},k)$-completion operator can be defined as the union of two unary operators $\overleftarrow{\mathrm{H}_k}$ and $\overrightarrow{\mathrm{H}_k}$.

\begin{definition}\label{def lang h left right}
  Let $\Gamma$ be an alphabet and $\mathrm{H}$ be an anti-morphism over $\Gamma^*$. Let $L$ be a language over $\Gamma$. Let $k>0$ be an integer. The right (resp. left) $(H,k)$-\emph{completion} of $L$ is the language $\overrightarrow{\mathrm{H}_k}(L)$ (resp. $\overleftarrow{\mathrm{H}_k}(L)$) defined by:
  
  \centerline{
    $\overrightarrow{\mathrm{H}_k}(L)=\{\alpha\beta\gamma\mathrm{H}(\beta)\mathrm{H}(\alpha)\mid\alpha,\beta,\gamma\in\Gamma^*\ \wedge\ \alpha\beta\gamma\mathrm{H}(\beta)\in L \wedge |\beta|=k\}$,
  }
  
  \centerline{
    $\overleftarrow{\mathrm{H}_k}(L)=\{\alpha\beta\gamma\mathrm{H}(\beta)\mathrm{H}(\alpha)\mid \alpha,\beta,\gamma\in\Gamma^*\ \wedge\ \beta\gamma\mathrm{H}(\beta)\mathrm{H}(\alpha)\in L \wedge |\beta|=k\}$.
  }
\end{definition}

\begin{lemma}\label{lem hkl1l2 op unaire}
  Let $\Gamma$ be an alphabet and $\mathrm{H}$ be an anti-morphism over $\Gamma^*$. Let $L_1$ and $L_2$ be two languages over $\Gamma$. Let $k>0$ be an integer. Then:
  
  \centerline{
    $\mathrm{H}_k(L_1,L_2)=\overrightarrow{\mathrm{H}_k}(L_1)\cup \overleftarrow{\mathrm{H}_k}(L_2)$.
  }
\end{lemma}
\begin{proof}
  Let $w$ be a word in $\Gamma^*$.
  
  \centerline{
    \begin{tabular}{l@{\ }l}
      $w\in \mathrm{H}_k(L_1,L_2)$ & $\Leftrightarrow$ 
        $
        \left\{
          \begin{array}{l}
            w=\alpha\beta\gamma \mathrm{H}(\beta)\mathrm{H}(\alpha)\\
            \wedge (\alpha\beta\gamma \mathrm{H}(\beta)\in L_1 \vee \beta\gamma \mathrm{H}(\beta)\mathrm{H}(\alpha)\in L_2)\\
            \wedge |\beta|=k\\
          \end{array}
        \right.$\\  
      & $\Leftrightarrow$ 
        $
        \left\{
          \begin{array}{l}
            (w=\alpha\beta\gamma \mathrm{H}(\beta)\mathrm{H}(\alpha) \wedge \alpha\beta\gamma \mathrm{H}(\beta)\in L_1\wedge\ |\beta|=k)\\
            \vee (w=\alpha\beta\gamma \mathrm{H}(\beta)\mathrm{H}(\alpha) \wedge \beta\gamma \mathrm{H}(\beta)\mathrm{H}(\alpha)\in L_2\wedge\ |\beta|=k)\\            
          \end{array}
        \right.$\\
      & $\Leftrightarrow$ $w\in \overrightarrow{\mathrm{H}_k}(L_1) \vee w\in \overleftarrow{\mathrm{H}_k}(L_2)$ $\Leftrightarrow$ $w\in \overrightarrow{\mathrm{H}_k}(L_1)\cup \overleftarrow{\mathrm{H}_k}(L_2)$.\\
    \end{tabular}
  }
    
  \cqfd
\end{proof}

When $\mathrm{H}$ is an involution over $\Gamma$, the $(\mathrm{H},k)$-completion of $L_1$ and $L_2$ is called a hairpin completion~\cite{CMM06}.
Even in the case where $\mathrm{H}$ is not an involution, we will say that languages such as $\overrightarrow{\mathrm{H}_k}(L)$, $\overleftarrow{\mathrm{H}_k}(L)$ or $\mathrm{H}_k(L,L')$ are hairpin completed languages and we will speak of hairpin completions.
We first establish formulae in this general setting in order to compute the two-sided residuals of the completed language of an arbitrary language. The following operator is useful.

\begin{definition}\label{def lang hprim}
  Let $\Gamma$ be an alphabet and $\mathrm{H}$ be an anti-morphism over $\Gamma^*$. Let $L$ be a language over an alphabet $\Gamma$. Let $k>0$ be an integer. The language $\mathrm{H}'_k(L)$ is defined by: $\mathrm{H}'_k(L)=\{\beta\gamma\mathrm{H}(\beta)\in L\mid\beta,\gamma\in\Gamma^*\ \wedge\ |\beta|=k\}.$
\end{definition}

  We split the computation of two-sided residuals of a completed language w.r.t. $(x,y)$ couples: the first case is when both $x$ and $y$ are symbols.

\begin{lemma}\label{lem mot ds lang h couple deb fin}
  Let $\Gamma$ be an alphabet and $\mathrm{H}$ be an anti-morphism over $\Gamma^*$. Let $L$ be a language over an alphabet $\Gamma$. Let $k>0$ be an integer.  Let $L'$ be a language in $\{\overleftarrow{\mathrm{H}_k}(L),\overrightarrow{\mathrm{H}_k}(L),\mathrm{H}'_k(L)\}$. Let $w$ a word in $\Gamma^*$. Then:
  
  \centerline{
    $w\in L'$ $\Rightarrow$ $|w|\geq k\ \wedge\ \exists a\in\Gamma,\exists w'\in\Gamma^*, w=aw'\mathrm{H}(a)$.
  }
\end{lemma}
\begin{proof}
  Trivially deduced from Definition~\ref{def lang h left right} and Definition~\ref{def lang hprim}.
  \cqfd
\end{proof}

\begin{corollary}\label{cor lang fonction der}
  Let $\Gamma$ be an alphabet and $\mathrm{H}$ be an anti-morphism over $\Gamma^*$. Let $L$ be a language over an alphabet $\Gamma$. Let $k>0$ be an integer. Let $L'$ be a language in $\{\overleftarrow{\mathrm{H}_k}(L),\overrightarrow{\mathrm{H}_k}(L),\mathrm{H}'_k(L)\}$. Then: $L'=\bigcup_{x\in \Gamma} \{x\}\cdot ((x,\mathrm{H}(x))^{-1}(L'))\cdot \{\mathrm{H}(x)\}$.
\end{corollary}

\begin{proposition}\label{prop lang quot hairpin}
  Let $\Gamma$ be an alphabet and $\mathrm{H}$ be an anti-morphism over $\Gamma^*$. Let $L$ be a language over $\Gamma$. Let $(x,y)$ a couple of symbols in $\Gamma\times\Gamma$. Let $k>0$ be an integer. Then:
  
  \centerline{
    $(x,y)^{-1}(\overrightarrow{\mathrm{H}_k}(L))=
      \left\{
        \begin{array}{l@{\ }l}
          \emptyset & \text{ if } y\neq \mathrm{H}(x),\\
          \overrightarrow{\mathrm{H}_k}(x^{-1}(L))\cup (x,y)^{-1}(L)  & \text{ if } y = \mathrm{H}(x)\ \wedge\ k=1,\\
          \overrightarrow{\mathrm{H}_k}(x^{-1}(L))\cup\mathrm{H}'_{k-1}((x,y)^{-1}(L)) & \text{ otherwise,}\\
        \end{array}
      \right.$
  }  
  
  \centerline{
    $(x,y)^{-1}(\overleftarrow{\mathrm{H}_k}(L))=
      \left\{
        \begin{array}{l@{\ }l}
          \emptyset & \text{ if } y \neq \mathrm{H}(x),\\
          \overleftarrow{\mathrm{H}_k}((L)y^{-1})\cup (x,y)^{-1}(L) & \text{ if } y = \mathrm{H}(x)\ \wedge\ k=1,\\
          \overleftarrow{\mathrm{H}_k}((L)y^{-1})\cup\mathrm{H}'_{k-1}((x,y)^{-1}(L)) & \text{ otherwise,}\\
        \end{array}
      \right.$
  }  
  
  \centerline{
    $(x,y)^{-1}(\mathrm{H}'_k(L))=
      \left\{
        \begin{array}{l@{\ }l}
          \emptyset & \text{ if } y \neq \mathrm{H}(x),\\
          \mathrm{H}'_{k-1}((x,y)^{-1}(L)) & \text{ if } y = \mathrm{H}(x)\ \wedge\ k>1,\\
          (x,y)^{-1}(L) & \text{ otherwise.}\\
        \end{array}
      \right.$
  }  
  
\end{proposition}
\begin{proof}
  Let $w$ be  a word in $\Gamma^*$. According to Lemma~\ref{lem mot ds lang h couple deb fin}, any word $u$ in $\overrightarrow{\mathrm{H}_k}(L)\cup \overleftarrow{\mathrm{H}_k}(x^{-1}(L))\cup \mathrm{H}'_k(L)$ can be split up into $avb$ with $b=\mathrm{H}(a)$. As a consequence, whenever $y\neq \mathrm{H}(x)$, it holds that $(x,y)^{-1}(\overrightarrow{\mathrm{H}_k}(L))=(x,y)^{-1}(\overleftarrow{\mathrm{H}_k}(L))=(x,y)^{-1}(\mathrm{H}'_k(L))=\emptyset$. Let us suppose now that $y= \mathrm{H}(x)$.
  
  \textbf{(I)} Let us define the languages $L_1$ and $L_2$ by:
   
   \centerline{$L_1=(x,y)^{-1}(\overrightarrow{\mathrm{H}_k}(L))$,}
   
   \centerline{$L_2=
    \left\{
      \begin{array}{l@{\ }l}
        \overrightarrow{\mathrm{H}_k}(x^{-1}(L))\cup\mathrm{H}'_{k-1}((x,y)^{-1}(L)) & \text{ if }k>1,\\
        \overrightarrow{\mathrm{H}_k}(x^{-1}(L))\cup (x,y)^{-1}(L)  & \text{ otherwise.}\\
      \end{array}
    \right.$}
    
    Then:
  
  \centerline{
    \begin{tabular}{l@{\ }l}
      $w\in L_1$ & $\Leftrightarrow$ $xwy \in \overrightarrow{\mathrm{H}_k}(L)$\\
      & $\Leftrightarrow$ 
  $
  \left\{
   \begin{array}{l}
    (
      xwy=x\alpha\beta\gamma\mathrm{H}(\beta)\mathrm{H}(\alpha)y 
      \wedge     
      y=\mathrm{H}(x)
      \wedge     
      x\alpha\beta\gamma\mathrm{H}(\beta) \in L
      \wedge 
      |\beta|=k    
    )\\
    \vee 
    (
      xwy=x\beta\gamma\mathrm{H}(\beta)y 
      \wedge     
      y=\mathrm{H}(x)
      \wedge     
      x\beta\gamma\mathrm{H}(\beta)y \in L
      \wedge 
      |\beta|=k-1 
    )\\
   \end{array}
   \right.
  $\\
  &
  $\Leftrightarrow$ 
  $
  \left\{
    \begin{array}{l}
    (
      w=\alpha\beta\gamma\mathrm{H}(\beta)\mathrm{H}(\alpha)
      \wedge     
      y=\mathrm{H}(x)
      \wedge     
      \alpha\beta\gamma\mathrm{H}(\beta) \in x^{-1}(L)
      \wedge 
      |\beta|=k    
    )\\
    \vee 
    (
      w=\beta\gamma\mathrm{H}(\beta) 
      \wedge     
      y=\mathrm{H}(x)
      \wedge     
      \beta\gamma\mathrm{H}(\beta) \in (x,y)^{-1}(L)
      \wedge 
      |\beta|=k-1 
    )\\
  \end{array}
  \right.
  $\\
  &
  $\Leftrightarrow$ 
  $
  \left\{
   \begin{array}{l}
    (
      w=\alpha\beta\gamma\mathrm{H}(\beta)\mathrm{H}(\alpha)
      \wedge     
      y=\mathrm{H}(x)
      \wedge     
      w\in \overrightarrow{\mathrm{H}_k}(x^{-1}(L))
    )\\
    \vee 
    (
      w=\beta\gamma\mathrm{H}(\beta) 
      \wedge     
      y=\mathrm{H}(x)
      \wedge     
      w \in \mathrm{H}'_{k-1}((x,y)^{-1}(L))
      \wedge 
      k\neq 1
    )\\
    \vee 
    (
      w= \gamma  
      \wedge     
      y=\mathrm{H}(x)
      \wedge     
      w \in  (x,y)^{-1}(L)
      \wedge 
      k= 1
    )\\
  \end{array}
  \right.
  $\\
  &
  $\Leftrightarrow$  $w\in L_2$.\\
  \end{tabular}}
  
  \textbf{(II)} Let us set:
  
  \centerline{ $L_1=(x,y)^{-1}(\overleftarrow{\mathrm{H}_k}(L))$,}
  
  \centerline{$L_2=
    \left\{
      \begin{array}{l@{\ }l}
        \overleftarrow{\mathrm{H}_k}(x^{-1}(L))\cup\mathrm{H}'_{k-1}((x,y)^{-1}(L)) & \text{ if }k>1,\\
        \overleftarrow{\mathrm{H}_k}(x^{-1}(L))\cup (x,y)^{-1}(L)  & \text{ otherwise.}\\
      \end{array}
    \right.$}
    
  Then
  
  \centerline{
    \begin{tabular}{l@{\ }l} 
      $w\in L_1$ & $\Leftrightarrow$ $xwy \in \overleftarrow{\mathrm{H}_k}(L)$\\
  & 
  $\Leftrightarrow$ 
  $
  \left\{
    \begin{array}{l}
    (
      xwy=x\alpha\beta\gamma\mathrm{H}(\beta)\mathrm{H}(\alpha)y 
      \wedge     
      y=\mathrm{H}(x)
      \wedge     
      \beta\gamma\mathrm{H}(\beta)\mathrm{H}(\alpha)y \in L
      \wedge 
      |\beta|=k    
    )\\
    \vee 
    (
      xwy=x\beta\gamma\mathrm{H}(\beta)y 
      \wedge     
      y=\mathrm{H}(x)
      \wedge     
      x\beta\gamma\mathrm{H}(\beta)y \in L
      \wedge 
      |\beta|=k-1 
    )\\
  \end{array}
  \right.
  $\\
  &  
  $\Leftrightarrow$ 
  $
  \left\{
   \begin{array}{l}
    (
      w=\alpha\beta\gamma\mathrm{H}(\beta)\mathrm{H}(\alpha)
      \wedge     
      y=\mathrm{H}(x)
      \wedge     
      \beta\gamma\mathrm{H}(\beta)\mathrm{H}(\alpha) \in (L)y^{-1}
      \wedge 
      |\beta|=k    
    )\\
    \vee 
    (
      w=\beta\gamma\mathrm{H}(\beta) 
      \wedge     
      y=\mathrm{H}(x)
      \wedge     
      \beta\gamma\mathrm{H}(\beta) \in (x,y)^{-1}(L)
      \wedge 
      |\beta|=k-1 
    )\\
  \end{array}
  \right.
  $ \\
  &
  $\Leftrightarrow$ 
  $\left\{
   \begin{array}{l}
    (
      w=\alpha\beta\gamma\mathrm{H}(\beta)\mathrm{H}(\alpha)
      \wedge     
      y=\mathrm{H}(x)
      \wedge     
      w\in \overleftarrow{\mathrm{H}_k}((L)y^{-1})
    )\\
    \vee 
    (
      w=\beta\gamma\mathrm{H}(\beta) 
      \wedge     
      y=\mathrm{H}(x)
      \wedge     
      w \in \mathrm{H}'_{k-1}((x,y)^{-1}(L))
      \wedge     
      k\neq 1
    )\\
    \vee 
    (
      w=\gamma
      \wedge     
      y=\mathrm{H}(x)
      \wedge     
      w \in (x,y)^{-1}(L)
      \wedge     
      k= 1
    )\\
   \end{array}
   \right.
  $\\
  &
  $\Leftrightarrow$  $w\in L_2$.\\
  \end{tabular}
  }

  \textbf{(III)} Let us set:
  
  \centerline{$L_1=(x,y)^{-1}(\mathrm{H}'_k(L))$,}
  
  \centerline{$L_2=\mathrm{H}'_{k-1}((x,y)^{-1}(L))$,}
  
  \centerline{$L_3=(x,y)^{-1}(L)$.}
  
  Then:
  
  \centerline{
    \begin{tabular}{l@{\ }l}
      $w\in L_1$ & $\Leftrightarrow$ $xwy \in \mathrm{H}'_k(L)$\\
      &
      $\Leftrightarrow$ 
  $\left\{
    \begin{array}{l}
      xwy=x\beta\gamma\mathrm{H}(\beta)y \\
      \wedge     
      y=\mathrm{H}(x)\\
      \wedge     
      x\beta\gamma\mathrm{H}(\beta)y \in L\\
      \wedge 
      |\beta|=k-1 \\
    \end{array}
  \right.
  $\\
  &
  $\Leftrightarrow$ 
  $\left\{
    \begin{array}{l}
      w=\beta\gamma\mathrm{H}(\beta)\\
      \wedge     
      y=\mathrm{H}(x)\\
      \wedge     
      \beta\gamma\mathrm{H}(\beta) \in (x,y)^{-1}(L)\\
      \wedge 
      |\beta|=k-1\\
    \end{array}
  \right. 
  $\\
  &
  $\Leftrightarrow$ 
  $\left\{
    \begin{array}{l}
    (
      w=\beta\gamma\mathrm{H}(\beta)
      \wedge     
      y=\mathrm{H}(x)
      \wedge     
      w \in \mathrm{H}'_{k-1}((x,y)^{-1}(L))
      \wedge 
      k>1 
    )\\
    \vee
    (
      w \in (x,y)^{-1}(L)
      \wedge 
      k=1 
    )\\
   \end{array}
   \right.
  $\\
  &
  $\Leftrightarrow$ 
  $\left\{
    \begin{array}{l}
    (
      w\in L_2
      \wedge 
      k>1 
    )\\
    \vee
    (
      w \in L_3
      \wedge 
      k=1 
    )\\
   \end{array}
   \right.
  $
  \end{tabular}
  }
  
  \cqfd
\end{proof}

The problem of two-sided residuals of an hairpin completion w.r.t. couples $(x,y)$ with either $x$ or $y$ equal to $\varepsilon$ is that they add one catenation that has to be memorized. It can be checked that this may lead to infinite sets of two-sided residuals.

\begin{proposition}\label{prop deriv x epsilon}
  Let $\Gamma$ be an alphabet and $\mathrm{H}$ be an anti-morphism over $\Gamma^*$. Let $L$ be a language over an alphabet $\Gamma$. Let $k>0$ be an integer. Let $L'$ be a language in $\{\overleftarrow{\mathrm{H}_k}(L),\overrightarrow{\mathrm{H}_k}(L),\mathrm{H}'_k(L)\}$. Let $x$ be a symbol in $\Gamma$. Then:
  
  \centerline{
    $(x,\varepsilon)^{-1}(L')= 
          (x,\mathrm{H}(x))^{-1}(L')\cdot \{\mathrm{H}(x)\}$,
  }
  
  \centerline{
    $(\varepsilon,x)^{-1}(L')=\bigcup_{z\in\Gamma \mid \mathrm{H}(z)=x} \{z\}\cdot (z,x)^{-1}(L')$.
  }  
\end{proposition}
\begin{proof}
  Directly deduced from Lemma~\ref{lem 2sidequot eq quot 2side} and from Corollary~\ref{cor lang fonction der}.
  \cqfd
\end{proof}

  Let $L$ be a language over an alphabet $\Gamma$. The set $\mathcal{R}_L$ \emph{of two-sided residuals} of $L$ is defined by: $\mathcal{R}_L=\bigcup_{k\geq 1} \mathcal{R}^k_L$, where
  
  \centerline{
    $\mathcal{R}^k_L=
      \left\{
        \begin{array}{l@{\ }l}
          \{(x,y)^{-1}(L)\mid (x,y)\in\Sigma_\Gamma\} & \text{ if } k=1,\\
          \{(x,y)^{-1}(L')\mid (x,y)\in\Sigma_\Gamma \wedge\ L'\in \mathcal{R}^{k-1}_L\} & \text{ otherwise.}\\
        \end{array}
      \right.
    $
  }    
  
From now on we focus on hairpin completion of regular languages. Let us recall that such a completion is not necessarily regular~\cite{CMM06}.

\begin{lemma}
  The family of regular languages is not closed under hairpin completion.\end{lemma}
\begin{proof}
  Let $\Gamma=\{a,b,c\}$, $k>0$ be a fixed integer and $\mathrm{H}$ be the anti-morphism over $\Gamma^*$ defined by $\mathrm{H}(a)=a$, $\mathrm{H}(b)=c$ and $\mathrm{H}(c)=b$. Let $L'=\overrightarrow{\mathrm{H}_k}(L(a^*b^kc^k))$.  
  Let us first show that $L'=\{a^nb^kc^ka^n\mid n\geq 0\}$. Let $w$ be a word in $\Gamma^*$.
  
  \centerline{
    \begin{tabular}{l@{\ }l}
      $w\in L'$ & $\Leftrightarrow$ $w=\alpha\beta\gamma\mathrm{H}(\beta)\mathrm{H}(\alpha) \wedge \alpha\beta\gamma\mathrm{H}(\beta)\in L(a^*b^kc^k) \wedge |\beta|=k$\\
      & $\Leftrightarrow$ $w=\alpha\beta\gamma\mathrm{H}(\beta)\mathrm{H}(\alpha) \wedge \alpha\in L(a^*)\wedge \mathrm{H}(\beta) =c^k \wedge \beta=b^k$\\
      & $\Leftrightarrow$ $w=a^n b^k c^k a^n$ with $n\geq 0$.\\
    \end{tabular}
  }
  
  For any integer $j\geq 0$, let us define the language $L'_j$ by:
  
  \centerline{$L'_j=
    \left\{
      \begin{array}{l@{\ }l}
        L' & \text{ if } j=0,\\ 
        a^{-1}(L'_{j-1}) & \text{ otherwise.}
      \end{array}
    \right.
  $}
  
  Consequently, it holds $L'_j=\{a^{n-j} b^k c^k a^n\mid n\geq j\}$.  
  Finally, since for any two distinct integers $j$ and $j'$, the word $b^k c^k a^j$ belongs to $L'_j\setminus L'_{j'}$, it holds that for any two distinct integers $j$ and $j'$, $L'_j\neq L'_{j'}$ and $(a^j)^{-1}(L')\neq (a^{j'})^{-1}(L')$.   
  As a consequence, the set of left residuals of $L'$ is infinite.
  \cqfd
\end{proof}

The set of two-sided residuals of a hairpin completion of a regular language may be infinite, but the restriction to residuals w.r.t. couples $(x,y)$ of symbols is sufficient to obtain a finite set of two-sided residuals and a finite recognizer.

\section{The Two-Sided Derived Term Automaton}\label{2sideddder}

  The computation of residuals is intractable when it is defined over languages. However, derived terms of regular expressions denote residuals of regular languages. We then extend the partial derivation of regular expressions~\cite{Ant96} to the partial derivation of hairpin expressions.
  
  A \emph{hairpin expression} $E$ over an alphabet $\Gamma$ is a regular expression over $\Gamma$ or is inductively defined by: $E=\overleftarrow{\mathrm{H}_k}(F)$, $E=\overrightarrow{\mathrm{H}_k}(F)$, $E=\mathrm{H}'_k(F)$, $E=G_1+G_2$, where $\mathrm{H}$ is any anti-morphism over $\Gamma^*$, $k>0$ is any integer, $F$ is any regular expression over $\Gamma$, and $G_1$ and $G_2$ are any two hairpin expressions over $\Sigma$. If the only operators appearing in $E$ are regular operators ($+$, $\cdot$ or $^*$), the expression $E$ is said to be a \emph{simple hairpin expression}.  The \emph{language} denoted by a hairpin expression $E$ over an alphabet $\Gamma$ is the regular language $L(E)$ if $E$ is a regular expression or is inductively defined by: $L(\overleftarrow{\mathrm{H}_k}(F))=\overleftarrow{\mathrm{H}_k}(L(F))$, $L(\overrightarrow{\mathrm{H}_k}(F))=\overrightarrow{\mathrm{H}_k}(L(F))$, $L(\mathrm{H}'_k(F))=\mathrm{H}'_k(L(F))$, $L(G_1+G_2)=L(G_1)\cup L(G_2)$,  where $\mathrm{H}$ is any anti-morphism over $\Gamma^*$, $k>0$ is any integer, $F$ is any regular expression over $\Gamma$, and $G_1$ and $G_2$ are any two hairpin expressions over $\Gamma$.

\begin{definition}\label{def form deriv part}
  Let $E$ be a hairpin expression over an alphabet $\Gamma$. Let $(x,y)$ be a couple of symbols in $\Sigma_\Gamma$. Let $k>0$ be an integer. The \emph{two-sided partial derivative} of $E$ \emph{w.r.t.} $(x,y)$ is the set $\frac{\partial}{\partial_{(x,y)}}(E)$ of hairpin expressions defined by:
  
  \centerline{
    $\frac{\partial}{\partial_{(x,y)}}(F)=
      \left\{
        \begin{array}{l@{\ }l}
          (F)\frac{\partial}{\partial_y} & \text{ if } x=\varepsilon,\\          \frac{\partial}{\partial_x}(F) & \text{ if } y=\varepsilon,\\          \bigcup_{F'\in\frac{\partial}{\partial_x}(F)}(F')\frac{\partial}{\partial_y} & \text{ otherwise,}\\
        \end{array}
      \right.
    $
  }  
  
  \centerline{
    $\frac{\partial}{\partial_{(x,y)}}(\overrightarrow{\mathrm{H}_k}(F))=      \left\{
        \begin{array}{l@{\ }l}
          \emptyset & \text{ if } y \neq \mathrm{H}(x),\\
          \overrightarrow{\mathrm{H}_k}(\frac{\partial}{\partial_{x}}(F))\cup \frac{\partial}{\partial_{(x,y)}}(F) & \text{ if }y = \mathrm{H}(x)\ \wedge\ k=1\\
          \overrightarrow{\mathrm{H}_k}(\frac{\partial}{\partial_{x}}(F))\cup \mathrm{H}'_{k-1}(\frac{\partial}{\partial_{(x,y)}}(F)) & \text{ otherwise,}\\
        \end{array}
      \right.$
  }  
  
  \centerline{
    $\frac{\partial}{\partial_{(x,y)}}(\overleftarrow{\mathrm{H}_k}(F))=      \left\{
        \begin{array}{l@{\ }l}
          \emptyset & \text{ if } y \neq \mathrm{H}(x),\\
          \overleftarrow{\mathrm{H}_k}((F)\frac{\partial}{\partial_{y}})\cup \frac{\partial}{\partial_{(x,y)}}(F) &  \text{ if }y = \mathrm{H}(x)\ \wedge\ k=1\\
          \overleftarrow{\mathrm{H}_k}((F)\frac{\partial}{\partial_{y}})\cup \mathrm{H}'_{k-1}(\frac{\partial}{\partial_{(x,y)}}(F)) & \text{ otherwise,}\\
        \end{array}
      \right.$
  }  
  
  \centerline{
    $\frac{\partial}{\partial_{(x,y)}}(\mathrm{H}'_k(F))=
      \left\{
        \begin{array}{l@{\ }l}
          \emptyset & \text{ if } y \neq \mathrm{H}(x),\\
          \mathrm{H}'_{k-1}(\frac{\partial}{\partial_{(x,y)}}(F)) & \text{ if k>1},\\
          \frac{\partial}{\partial_{(x,y)}}(F) & \text{ otherwise,}\\
        \end{array}
      \right.$
  }

  \centerline{
    $\frac{\partial}{\partial_{(x,y)}}(G_1+G_2)= \frac{\partial}{\partial_{(x,y)}}(G_1) \cup \frac{\partial}{\partial_{(x,y)}}(G_2) $,
  }

  where $\mathrm{H}$ is any anti-morphism over $\Gamma^*$, $k>0$ is any integer, $F$ is any regular expression over $\Gamma$, $G_1$ and $G_2$ are any two hairpin expressions over $\Gamma$, and for any set $\mathcal{H}$ of hairpin expressions: $\overrightarrow{\mathrm{H}_k}(\mathcal{H})= \{\overrightarrow{\mathrm{H}_k}(H)\mid H\in \mathcal{H}\}$, $\overleftarrow{\mathrm{H}_k}(\mathcal{H})=\{\overleftarrow{\mathrm{H}_k}(H)\mid H\in \mathcal{H}\}$, $\mathrm{H}'_k(\mathcal{H})=\{\mathrm{H}'_k(H)\mid H\in \mathcal{H}\}$.
\end{definition}

  Let $E$ be a hairpin expression over an alphabet $\Gamma$. The \emph{set} $\overleftrightarrow{\mathcal{D}_E}$ \emph{of two-sided derived terms of the expression} $E$ is defined by: $\overleftrightarrow{\mathcal{D}_E}=\bigcup_{k\geq 1} \overleftrightarrow{\mathcal{D}^k_E}$, where:
  
  \centerline{
    $\overleftrightarrow{\mathcal{D}^k_E}=
      \left\{
        \begin{array}{l@{\ }l}
          \bigcup_{(x,y)\in\Sigma_\Gamma} \frac{\partial}{\partial(x,y)}(E)  & \text{ if } k=1,\\
          \bigcup_{(x,y)\in\Sigma_\Gamma,E'\in \overleftrightarrow{\mathcal{D}^{k-1}_E}} \frac{\partial}{\partial(x,y)}(E') & \text{ otherwise.}\\        \end{array}
      \right.
    $
  }

Derived terms of regular expressions are related to left residuals. Let us show that derived terms of hairpin expressions are related to two-sided residuals.

\begin{proposition}\label{prop deriv two side bon lang}
Let $E$ be a hairpin expression over an alphabet $\Gamma$. Let $(x,y)$ be a couple of symbols in $\Gamma^2$. Then: $\bigcup_{F\in \frac{\partial}{\partial_{(x,y)}}(E)}L(F)=(x,y)^{-1}(L(E))$.
  
  Furthermore, if $E$ is a regular expression, the proposition still holds whenever $(x,y)$ is a couple of symbols in $\Sigma_\Gamma$.
\end{proposition}
\begin{proof}
  Trivially proved by induction over the structure of $E$, according to Proposition~\ref{prop lang quot hairpin}.
  \cqfd
\end{proof}

Determining whether the empty word belongs to the language denoted by a regular expression $E$ can be performed syntactically and inductively as follows:  

\centerline{$\varepsilon\notin L(a)$, $\varepsilon\notin L(\emptyset)$, $\varepsilon\in L(\varepsilon)$,}

\centerline{$\varepsilon\in L(G_1\cdot G_2) \Leftrightarrow \varepsilon\in L(G_1)\ \wedge\ \varepsilon\in L(G_2)$,}

\centerline{$\varepsilon\in L(G_1+G_2) \Leftrightarrow \varepsilon\in L(G_1)\ \vee\ \varepsilon\in L(G_2)$, $\varepsilon\in L(G_1^*)$.}

This syntactical test is needed to compute the derived term automaton since it defines the finality of the states. We now show how to extend this computation to hairpin expressions.

\begin{lemma}\label{lem test syntax eps}
  Let $F$ be a regular expression and $G_1$ and $G_2$ be two hairpin expressions. Then:
  
\centerline{ $\varepsilon\notin L(\overrightarrow{\mathrm{H}_k}(F))$, $\varepsilon\notin L(\overleftarrow{\mathrm{H}_k}(F))$, $\varepsilon\notin L(\mathrm{H}'_k(F))$,}

\centerline{ $\varepsilon\in L(G_1+G_2) \Leftrightarrow \varepsilon\in L(G_1)\ \vee\ \varepsilon\in L(G_2)$.}
\end{lemma}
\begin{proof}
  Trivially proved according to Definition~\ref{def lang h left right}, Definition~\ref{def lang hprim} and definition of languages denoted by hairpin expressions.  \cqfd
\end{proof}

The following example illustrates the computation of derived terms. For clarity, in this example, we assume that hairpin expressions are quotiented w.r.t. the following rules: $\varepsilon\cdot E\sim E$, $\emptyset\cdot E\sim \emptyset$. Moreover, sets of expressions are also quotiented w.r.t. the following rule: $\{\emptyset\}\sim\emptyset$.
  
\begin{example}\label{ex calcul deriv term}
  Let $\Gamma=\{a,b,c\}$ and $\mathrm{H}$ be the anti-morphism over $\Gamma^*$ defined by $\mathrm{H}(a)=a$, $\mathrm{H}(b)=c$ and $\mathrm{H}(c)=b$. Let $E=\overrightarrow{\mathrm{H}_1}(a^*bc)$. Derived terms of $E$ are computed as follows:
  
   \centerline{$\frac{\partial}{\partial(a,a)}(E)=\{E\}$,}
   
     \centerline{$\frac{\partial}{\partial(b,c)}(E)=\{\overrightarrow{\mathrm{H}_1}(c),\varepsilon\}$,}
    
      \centerline{$\frac{\partial}{\partial(c,b)}(\overrightarrow{\mathrm{H}_1}(c))=\{\overrightarrow{\mathrm{H}_1}(\varepsilon)\}$.}
     
   Other partial derivatives are equal to $\emptyset$.  
  Furthermore, it holds that $\varepsilon$ is the only derived term $F$ of $E$ such that $\varepsilon$ belongs to $L(F)$.
\end{example}

In the following we are looking for an upper bound over the cardinality of the set of two-sided derived terms, thus we apply no reduction to the regular expressions. Notice that this cardinality decreases whenever any reduction is applied.

\begin{lemma}\label{lem form ind calcul deriv 2side}
  Let $E$ and $F$ be two regular expressions over an alphabet $\Gamma$. Then the three following propositions hold:
  \begin{enumerate}
    \item $\overleftrightarrow{\mathcal{D}_{E+F}}\subset \overleftrightarrow{\mathcal{D}_E} \cup \overleftrightarrow{\mathcal{D}_F}$,
    \item $\overleftrightarrow{\mathcal{D}_{E\cdot F}}\subset  \overleftarrow{\mathcal{D}_E}\cdot \overrightarrow{\mathcal{D}_F} \cup \overleftrightarrow{\mathcal{D}_{E}}\cup \overleftrightarrow{\mathcal{D}_{F}}$,
    \item $\overleftrightarrow{\mathcal{D}_{E^*}}\subset \overleftarrow{\mathcal{D}_{E}}\cdot E^* \cup E^*\cdot \overrightarrow{\mathcal{D}_{E}}\cup \overleftrightarrow{\mathcal{D}_{E}}\cup (\overleftarrow{\mathcal{D}_{E}}\cdot E^*) \cdot \overrightarrow{\mathcal{D}_{E}}\cup \overleftarrow{\mathcal{D}_{E}}\cdot( E^* \cdot \overrightarrow{\mathcal{D}_{E}})$.
  \end{enumerate}
  
  Furthermore, $\overleftrightarrow{\mathcal{D}_{\varepsilon}}=\overleftrightarrow{\mathcal{D}_{\emptyset}}=\emptyset$ and $\overleftrightarrow{\mathcal{D}_{a}}=\{\varepsilon\}$ for any symbol $a$ in $\Gamma$.
\end{lemma}
\begin{proof}
  Basic cases ($\varepsilon$, $\emptyset$ and $a$ in $\Gamma$) are trivially proved directly applying Definition~\ref{def form deriv part}.
  
  By induction over the structure of the set of two-sided derived terms. Suppose that $E$ and $F$ are two regular expressions over an alphabet $\Gamma$. Let $(x,y)$ be a couple of symbols in $\Sigma_\Gamma$.
  \begin{enumerate}
    \item Let us first show that $\frac{\partial}{\partial(x,y)}(E+F)\subset \overleftrightarrow{\mathcal{D}_E} \cup \overleftrightarrow{\mathcal{D}_F}$. According to Definition~\ref{def form deriv part}, it holds:
    
    \centerline{
      $\frac{\partial}{\partial(x,y)}(E+F)=
        \left\{
          \begin{array}{l@{\ }l}
            \frac{\partial}{\partial_x}(E+F) & \text{ if }y=\varepsilon,\\
            (E+F)\frac{\partial}{\partial_y} & \text{ if }x=\varepsilon,\\
            \bigcup_{G\in \frac{\partial}{\partial_x}(E+F)} (G)\frac{\partial}{\partial_y} & \text{ otherwise.}\\
          \end{array}
        \right.
      $
    }
    
    \centerline{
      $=
        \left\{
          \begin{array}{l@{\ }l}
            \frac{\partial}{\partial_x}(E)\cup \frac{\partial}{\partial_x}(F) & \text{ if }y=\varepsilon,\\
            (E)\frac{\partial}{\partial_y}\cup (F)\frac{\partial}{\partial_y} & \text{ if }x=\varepsilon,\\
            \bigcup_{G\in \frac{\partial}{\partial_x}(E)} (G)\frac{\partial}{\partial_y} \cup \bigcup_{G\in \frac{\partial}{\partial_x}(F)} (G)\frac{\partial}{\partial_y} & \text{ otherwise.}\\
          \end{array}
        \right.
      $
    }
    
    Notice that the three following conditions hold:
    
    \centerline{$\frac{\partial}{\partial_x}(E)\cup \frac{\partial}{\partial_x}(F)\subset \overleftarrow{\mathcal{D}_E}\cup \overleftarrow{\mathcal{D}_F}\subset \overleftrightarrow{\mathcal{D}_E} \cup \overleftrightarrow{\mathcal{D}_F}$,}
    
    \centerline{$(E)\frac{\partial}{\partial_y}\cup (F)\frac{\partial}{\partial_y} \subset \overrightarrow{\mathcal{D}_E}\cup \overrightarrow{\mathcal{D}_F}\subset \overleftrightarrow{\mathcal{D}_E} \cup \overleftrightarrow{\mathcal{D}_F}$,}
    
    \centerline{$\bigcup_{G\in \frac{\partial}{\partial_x}(E)} (G)\frac{\partial}{\partial_y} \cup \bigcup_{G\in \frac{\partial}{\partial_x}(F)} (G)\frac{\partial}{\partial_y}= \frac{\partial}{\partial(x,y)}(E)\cup \frac{\partial}{\partial(x,y)}(F) \subset \overleftrightarrow{\mathcal{D}_E} \cup \overleftrightarrow{\mathcal{D}_F}$.}
    
    As a consequence, $\frac{\partial}{\partial(x,y)}(E+F)\subset \overleftrightarrow{\mathcal{D}_E} \cup \overleftrightarrow{\mathcal{D}_F}$. 
    
    Furthermore, since by definition of the sets of two-sided derived terms, for any expression $G$ in $\overleftrightarrow{\mathcal{D}_E}$ (resp. in $\overleftrightarrow{\mathcal{D}_F}$), $\frac{\partial}{\partial(x,y)}(G)\subset \overleftrightarrow{\mathcal{D}_E}$ (resp. $\frac{\partial}{\partial(x,y)}(G)\subset \overleftrightarrow{\mathcal{D}_F}$), the proposition is satisfied.
    \item Let us set $\mathcal{E}= \overleftarrow{\mathcal{D}_E}\cdot \overrightarrow{\mathcal{D}_F}\cup \overleftrightarrow{\mathcal{D}_{E}}\cup \overleftrightarrow{\mathcal{D}_{F}}$.
    \begin{enumerate}    
    \item Let us first show that 
    
    \centerline{$\frac{\partial}{\partial(x,y)}(E\cdot F)\subset
     \frac{\partial}{\partial_x}(E) \cdot (F)\frac{\partial}{\partial_y}\cup \frac{\partial}{\partial(x,y)}(E) \cup \frac{\partial}{\partial(x,y)}(F)
    \subset \mathcal{E}$.}
    
     According to Definition~\ref{def form deriv part}, it holds:
    
    \centerline{
      $\frac{\partial}{\partial(x,y)}(E\cdot F)=
        \left\{
          \begin{array}{l@{\ }l}
            \frac{\partial}{\partial_x}(E\cdot F) & \text{ if }y=\varepsilon,\\
            (E\cdot F)\frac{\partial}{\partial_y} & \text{ if }x=\varepsilon,\\
            \bigcup_{G\in \frac{\partial}{\partial_x}(E\cdot F)} (G)\frac{\partial}{\partial_y} & \text{ otherwise.}\\
          \end{array}
        \right.
      $
    }
    
    \centerline{$=
        \left\{
          \begin{array}{l@{\ }l}
            \frac{\partial}{\partial_x}(E)\cdot F \cup \frac{\partial}{\partial_x}(F) & \text{ if }y=\varepsilon\ \wedge \varepsilon\in L(E),\\
            \frac{\partial}{\partial_x}(E)\cdot F & \text{ if }y=\varepsilon\ \wedge \varepsilon\notin L(E),\\
             (E)\frac{\partial}{\partial_y} \cup E\cdot(F)\frac{\partial}{\partial_y} & \text{ if }x=\varepsilon\ \wedge \varepsilon\in L(F),\\
              E\cdot(F)\frac{\partial}{\partial_y} & \text{ if }x=\varepsilon\ \wedge \varepsilon\notin L(F),\\
            \bigcup_{G\in \frac{\partial}{\partial_x}(E)\cdot F} (G)\frac{\partial}{\partial_y} \cup \bigcup_{G\in \frac{\partial}{\partial_x}(F)} (G)\frac{\partial}{\partial_y} & \text{ if }x,y\in\Gamma\ \wedge\ \varepsilon\in L(E)\\
            \bigcup_{G\in \frac{\partial}{\partial_x}(E)\cdot F} (G)\frac{\partial}{\partial_y}  & \text{ otherwise.}\\
          \end{array}
        \right.
      $
    }
    
    Notice that the three following conditions hold:
    
    \centerline{$\frac{\partial}{\partial_x}(E)\cdot F \cup \frac{\partial}{\partial_x}(F) \subset \overleftarrow{\mathcal{D}_E}\cdot \overrightarrow{\mathcal{D}_F}\cup \overleftrightarrow{\mathcal{D}_{F}}$,}
    
    \centerline{$(E)\frac{\partial}{\partial_y} \cup E\cdot(F)\frac{\partial}{\partial_y}\subset \overleftrightarrow{\mathcal{D}_{E}} \cup \overleftarrow{\mathcal{D}_E} \cdot  \overrightarrow{\mathcal{D}_F}$,}
    
    \centerline{$\bigcup_{G\in \frac{\partial}{\partial_x}(F)} (G)\frac{\partial}{\partial_y}=\frac{\partial}{\partial(x,y)}(F)\subset \overrightarrow{\mathcal{D}_F}$.}
    
    Moreover,
    
    \centerline{
      $\bigcup_{G\in \frac{\partial}{\partial_x}(E)\cdot F} (G)\frac{\partial}{\partial_y}=
        \left\{
          \begin{array}{l@{\ }l}
            \bigcup_{G \in \frac{\partial}{\partial_x}(E)} G\cdot (F)\frac{\partial}{\partial_y} \cup \bigcup_{G\in \frac{\partial}{\partial_x}(E)} (G)\frac{\partial}{\partial_y} & \text{ if }\varepsilon\in L(F),\\
            \bigcup_{G\in \frac{\partial}{\partial_x}(E)} G\cdot (F)\frac{\partial}{\partial_y} & \text{ otherwise.}\\
          \end{array}
        \right.
      $
    }
    
    Finally, since $\bigcup_{G \in \frac{\partial}{\partial_x}(E)} G\cdot (F)\frac{\partial}{\partial_y}=\frac{\partial}{\partial_x}(E) \cdot (F)\frac{\partial}{\partial_y}\subset \overleftarrow{\mathcal{D}_E}\cdot \overrightarrow{\mathcal{D}_F}$ and since $\bigcup_{G\in \frac{\partial}{\partial_x}(E)} (G)\frac{\partial}{\partial_y}=\frac{\partial}{\partial(x,y)}(E)\subset \overrightarrow{\mathcal{D}_E}$, the proposition is satisfied.    
    \item Let us now show that for any expression $G$ in $\mathcal{E}$, $\frac{\partial}{\partial(x,y)}(G)\subset \mathcal{E}$.
    \begin{enumerate}
      \item if $G$ belongs to $\overleftrightarrow{\mathcal{D}_{E}}$ (resp. to $\overleftrightarrow{\mathcal{D}_{F}}$), by definition of the set of two-sided derived terms it holds $\frac{\partial}{\partial(x,y)}(G)\subset \overleftrightarrow{\mathcal{D}_{E}}$ (resp. $\frac{\partial}{\partial(x,y)}(G)\subset \overleftrightarrow{\mathcal{D}_{F}}$).
      \item If $G$ belongs to $\overleftarrow{\mathcal{D}_E}\cdot \overrightarrow{\mathcal{D}_F}$, then $G=G_1\cdot G_2$ and from \textbf{(2a)} it holds that $\frac{\partial}{\partial(x,y)}(G)\subset \overleftarrow{\mathcal{D}_{G_1}}\cdot \overrightarrow{\mathcal{D}_{G_2}}\cup \overleftrightarrow{\mathcal{D}_{G_1}}\cup \overleftrightarrow{\mathcal{D}_{{G_2}}}$. According to definition of the set of two-sided derived terms, the four follwong conditions hold:
      
       $\overleftarrow{\mathcal{D}_{G_1}}\subset \overleftarrow{\mathcal{D}_{E}}$,  $\overleftrightarrow{\mathcal{D}_{G_1}}\subset \overleftrightarrow{\mathcal{D}_{E}}$,  $\overrightarrow{\mathcal{D}_{G_2}}\subset \overrightarrow{\mathcal{D}_{F}}$ and   $\overleftrightarrow{\mathcal{D}_{G_2}}\subset \overleftrightarrow{\mathcal{D}_{F}}$.      
    \end{enumerate}
    \end{enumerate}
    As a consequence, the proposition is satisfied.
    \item Let us set $\mathcal{E}=\overleftarrow{\mathcal{D}_{E}}\cdot E^* \cup E^*\cdot \overrightarrow{\mathcal{D}_{E}}\cup \overleftrightarrow{\mathcal{D}_{E}}\cup (\overleftarrow{\mathcal{D}_{E}}\cdot E^*) \cdot \overrightarrow{\mathcal{D}_{E}}\cup \overleftarrow{\mathcal{D}_{E}}\cdot( E^* \cdot \overrightarrow{\mathcal{D}_{E}})$.
    \begin{enumerate}
      \item Let us first show that $\frac{\partial}{\partial(x,y)}(E^*)\subset \mathcal{E}$. According to Definition~\ref{def form deriv part}, it holds:
    
    \centerline{
      $\frac{\partial}{\partial(x,y)}(E^*)=
        \left\{
          \begin{array}{l@{\ }l}
            \frac{\partial}{\partial_x}(E^*) & \text{ if }y=\varepsilon,\\
            (E^*)\frac{\partial}{\partial_y} & \text{ if }x=\varepsilon,\\
            \bigcup_{G\in \frac{\partial}{\partial_x}(E^*)} (G)\frac{\partial}{\partial_y} & \text{ otherwise.}\\
          \end{array}
        \right.
      $
    }
    
    \centerline{
      $=
        \left\{
          \begin{array}{l@{\ }l}
            \frac{\partial}{\partial_x}(E)\cdot E^* & \text{ if }y=\varepsilon,\\
            E^*\cdot (E)\frac{\partial}{\partial_y} & \text{ if }x=\varepsilon,\\
            \bigcup_{G\in \frac{\partial}{\partial_x}(E)} (G\cdot E^*)\frac{\partial}{\partial_y} & \text{ otherwise.}\\
          \end{array}
        \right.
      $
    }
    
    Notice that $\frac{\partial}{\partial_x}(E)\cdot E^* \subset \overleftarrow{\mathcal{D}_{E}}\cdot E^*$ and that $E^*\cdot (E)\frac{\partial}{\partial_y} \subset E^*\cdot \overrightarrow{\mathcal{D}_{E}}$.
    
    Moreover,
    
    \centerline{
      $\bigcup_{G\in \frac{\partial}{\partial_x}(E)} (G\cdot E^*)\frac{\partial}{\partial_y}$
    }
    
    \centerline{$=
      \bigcup_{G\in \frac{\partial}{\partial_x}(E)} (G)\frac{\partial}{\partial_y} \cup G\cdot (E^*)\frac{\partial}{\partial_y}
    $
    }
    
    \centerline{$=
      \bigcup_{G\in \frac{\partial}{\partial_x}(E)} (G)\frac{\partial}{\partial_y} \cup G\cdot (E^*\cdot (E)\frac{\partial}{\partial_y})
    $
    }
    
    \centerline{$=
      \bigcup_{G\in \frac{\partial}{\partial_x}(E)} (G)\frac{\partial}{\partial_y} \cup  \bigcup_{G\in \frac{\partial}{\partial_x}(E)} G\cdot (E^*\cdot (E)\frac{\partial}{\partial_y})
    $
    }
    
    Finally, since the two following conditions hold:
    
    \centerline{$\bigcup_{G\in \frac{\partial}{\partial_x}(E)} (G)\frac{\partial}{\partial_y}=\frac{\partial}{\partial(x,y)}(E)\subset \overleftrightarrow{\mathcal{D}_{E}}$}
    
    \centerline{ and $\bigcup_{G\in \frac{\partial}{\partial_x}(E)} G\cdot (E^*\cdot (E)\frac{\partial}{\partial_y})=\frac{\partial}{\partial_x}(E)\cdot (E^*\cdot (E)\frac{\partial}{\partial_y})\subset \overleftarrow{\mathcal{D}_{E}}\cdot( E^* \cdot \overrightarrow{\mathcal{D}_{E}})$,}
    
    it holds that $\frac{\partial}{\partial(x,y)}(E^*)\subset \mathcal{E}$.    
      \item Let us now show that for any expression $G$ in $\mathcal{E}$, $\frac{\partial}{\partial(x,y)}(G)\subset\mathcal{E}$.
      \begin{enumerate}
        \item if $G$ belongs to $\overleftrightarrow{\mathcal{D}_{E}}$, by definition of the set of two-sided derived terms it holds $\frac{\partial}{\partial(x,y)}\subset \overleftrightarrow{\mathcal{D}_{E}}$.
        \item if $G$ belongs to $\overleftarrow{\mathcal{D}_{E}}\cdot E^*$, then $G=G_1\cdot E^*$ and from \textbf{(2a)} it holds that:
        
        \centerline{ $\frac{\partial}{\partial(x,y)}(G)\subset 
         \frac{\partial}{\partial_x}({G_1}) \cdot ({E^*})\frac{\partial}{\partial_y}\cup \frac{\partial}{\partial(x,y)}({G_1}) \cup \frac{\partial}{\partial(x,y)}({E^*})
        $.}
        
         Moreover, since from \textbf{(3a)} $\frac{\partial}{\partial(x,y)}({E^*})\subset\mathcal{E}$, since $\frac{\partial}{\partial_x}({G_1}) \subset \overleftarrow{\mathcal{D}_{E}}$ and since $\frac{\partial}{\partial(x,y)}({G_1}) \subset \overleftrightarrow{\mathcal{D}_{E}}$, it holds that:
        
        \centerline{
        $\frac{\partial}{\partial(x,y)}(G) $
        }
        
        \centerline{$\subset        
        \overleftarrow{\mathcal{D}_{E}} \cdot (E^*\cdot ({E})\frac{\partial}{\partial_y})
        \cup \overleftrightarrow{\mathcal{D}_{E}}
        \cup \mathcal{E}
        $
        }
        
        \centerline{$\subset
        \overleftarrow{\mathcal{D}_{E}} \cdot (E^*\cdot \overrightarrow{\mathcal{D}_{E}})
        \cup \overleftrightarrow{\mathcal{D}_{E}}
        \cup \mathcal{E}
        $
        }
        
        \centerline{$\subset \mathcal{E}
        $
        }

        \item if $G$ belongs to $E^*\cdot \overrightarrow{\mathcal{D}_{E}}$, then $G=E^* \cdot G_1$ and from \textbf{(2a)} it holds that:
        
        \centerline{ $\frac{\partial}{\partial(x,y)}(G)\subset 
       \frac{\partial}{\partial_x}({E^*}) \cdot ({G_1})\frac{\partial}{\partial_y}\cup \frac{\partial}{\partial(x,y)}({E^*}) \cup \frac{\partial}{\partial(x,y)}({G_1})
        $.}
        
         Moreover, since from \textbf{(3a)} $\frac{\partial}{\partial(x,y)}({E^*})\subset\mathcal{E}$, since $({G_1})\frac{\partial}{\partial_y} \subset \overrightarrow{\mathcal{D}_{E}}$ and since $\frac{\partial}{\partial(x,y)}({G_1}) \subset \overleftrightarrow{\mathcal{D}_{E}}$:
        
        \centerline{
        $\frac{\partial}{\partial(x,y)}(G) $
        }
        
        \centerline{$\subset
        (\frac{\partial}{\partial_x}({E})\cdot E^*) \cdot \overrightarrow{\mathcal{D}_{E}}
        \cup \mathcal{E}
        \cup \overleftrightarrow{\mathcal{D}_{E}}
        $
        }
        
        \centerline{$\subset
        (\overleftarrow{\mathcal{D}_{E}}\cdot E^*) \cdot \overrightarrow{\mathcal{D}_{E}}
        \cup \mathcal{E}
        \cup \overleftrightarrow{\mathcal{D}_{E}}
        $
        }
        
        \centerline{$\subset \mathcal{E}
        $
        }

        \item If $G$ belongs to $(\overleftarrow{\mathcal{D}_{E}}\cdot E^*) \cdot \overrightarrow{\mathcal{D}_{E}}$, then $G=(G_1\cdot E^*)\cdot G_2$ and from \textbf{(2a)} it holds that:
        
        \centerline{ $\frac{\partial}{\partial(x,y)}(G)\subset 
        \frac{\partial}{\partial_x}({G_1\cdot E^*}) \cdot ({G_2})\frac{\partial}{\partial_y}
        \cup \frac{\partial}{\partial(x,y)}({G_1\cdot E^*}) 
        \cup \frac{\partial}{\partial(x,y)}({G_2})$.}
        
        Since $\frac{\partial}{\partial_x}({G_1\cdot E^*})\subset \frac{\partial}{\partial_x}({G_1}) \cdot E^* \cup \frac{\partial}{\partial_x}({E^*})$, it holds that:
        
        \centerline{ $\frac{\partial}{\partial_x}({G_1\cdot E^*}) \cdot ({G_2})\frac{\partial}{\partial_y}\subset (\frac{\partial}{\partial_x}({G_1}) \cdot E^*)\cdot ({G_2})\frac{\partial}{\partial_y} \cup (\frac{\partial}{\partial_x}({E})\cdot E^*) \cdot ({G_2})\frac{\partial}{\partial_y}$.}
        
                Finally, since from \textbf{(3bii)} $\frac{\partial}{\partial(x,y)}({G_1\cdot E^*})\subset \mathcal{E}$, it holds:
        
        \centerline{
          $\frac{\partial}{\partial(x,y)}(G)$}
          
        \centerline{$
        \subset 
        (\overleftarrow{\mathcal{D}_E} \cdot E^*)\cdot \overrightarrow{\mathcal{D}_E} 
        \cup \mathcal{E} 
        \cup \overleftrightarrow{\mathcal{D}_E}$
        }
          
        \centerline{$
        \subset  \mathcal{E}  $
        }
        \item If $G$ belongs to $\overleftarrow{\mathcal{D}_{E}}\cdot( E^* \cdot \overrightarrow{\mathcal{D}_{E}})$, then $G=G_1\cdot (E^*\cdot G_2)$ and from \textbf{(2a)} it holds that:
        
        \centerline{ $\frac{\partial}{\partial(x,y)}(G)\subset 
        \frac{\partial}{\partial_x}({G_1}) \cdot ({E^* \cdot G_2})\frac{\partial}{\partial_y}
        \cup \frac{\partial}{\partial(x,y)}({G_1}) 
        \cup \frac{\partial}{\partial(x,y)}({E^* \cdot G_2})$.}
        
         Since $({E^* \cdot G_2})\frac{\partial}{\partial_y}\subset ({E^* })\frac{\partial}{\partial_y} \cup E^* \cdot(G_2)\frac{\partial}{\partial_y}$, it holds that:
         
         \centerline{  $\frac{\partial}{\partial_x}({G_1}) \cdot ({E^* \cdot G_2})\frac{\partial}{\partial_y}\subset \frac{\partial}{\partial_x}({G_1})\cdot ({E^* })\frac{\partial}{\partial_y} \cup \frac{\partial}{\partial_x}({G_1})\cdot(E^* \cdot(G_2)\frac{\partial}{\partial_y})$.}
         
          Finally, since from \textbf{(3biii)} $\frac{\partial}{\partial(x,y)}({E^* \cdot G_2})\subset \mathcal{E}$, it holds:
        
        \centerline{
          $\frac{\partial}{\partial(x,y)}(G)$}
          
        \centerline{$
        \subset 
        \overleftarrow{\mathcal{D}_E} \cdot (E^*\cdot \overrightarrow{\mathcal{D}_E} )
        \cup \mathcal{E} 
        \cup \overleftrightarrow{\mathcal{D}_E}$
        }
          
        \centerline{$
        \subset  \mathcal{E}  $
        }
      \end{enumerate}
    \end{enumerate}
    As a consequence, the proposition is satisfied.
  \end{enumerate}
  \cqfd
\end{proof}

\begin{proposition}\label{prop nbre 2side quot reg exp}
  Let $E$ be a regular expression of width $n>0$ and of star number $h$. Let us set $m=n+h$. Then the three following propositions hold:
  \begin{enumerate}
    \item $\mathrm{Card}(\overleftarrow{\mathcal{D}_E})\leq n$,
    \item $\mathrm{Card}(\overrightarrow{\mathcal{D}_E})\leq n$,
    \item $\mathrm{Card}(\overleftrightarrow{\mathcal{D}_E})\leq \frac{2 m\times (m+1) \times (m+2)}{3}-3$.
  \end{enumerate} 
\end{proposition}
\begin{proof}
  For the set of left derived terms, the proposition is proved in~\cite{Ant96}, where it is shown that the cardinality of the set $\{E'\mid \exists w\in\Sigma^+, E'\in\frac{\partial}{\partial_w}(E)\}$ is less than $n$. This bound still holds for the set of right derived terms.
  
  Let $n_1$ (resp. $n_2$) be the width of a regular expression $F$ (resp. $G$) and $h_1$ (resp. $h_2$) be the star number of $F$ (resp. $G$). Let us set $m_1=n_1+h_1$ and $m_2=n_2+h_2$.  
  For $E=F+G$ and for $E=F\cdot G$, we have $n=n_1+n_2$, $h=h_1+h_2$ and $m=m_1+m_2$. For $E=F^*$, we have $n=n_1$, $h=h_1+1$ and $m=m_1+1$.
  
  According to  Lemma~\ref{lem form ind calcul deriv 2side}, we get:
  \begin{enumerate}
    \item $\overleftrightarrow{\mathcal{D}_{F+G}}\subset \overleftrightarrow{\mathcal{D}_F} \cup \overleftrightarrow{\mathcal{D}_G}$,
    \item $\overleftrightarrow{\mathcal{D}_{F\cdot G}}\subset  \overleftarrow{\mathcal{D}_F}\cdot \overrightarrow{\mathcal{D}_G} \cup \overleftrightarrow{\mathcal{D}_{F}}\cup \overleftrightarrow{\mathcal{D}_{G}}$,
    \item $\overleftrightarrow{\mathcal{D}_{F^*}}\subset \overleftarrow{\mathcal{D}_{F}}\cdot F^* \cup F^*\cdot \overrightarrow{\mathcal{D}_{F}}\cup \overleftrightarrow{\mathcal{D}_{F}}\cup (\overleftarrow{\mathcal{D}_{F}}\cdot F^*) \cdot \overrightarrow{\mathcal{D}_{F}}\cup \overleftarrow{\mathcal{D}_{F}}\cdot( F^* \cdot \overrightarrow{\mathcal{D}_{F}})$.
  \end{enumerate}
  
  As a consequence, we get:
  \begin{enumerate}
    \item $\mathrm{Card}(\overleftrightarrow{\mathcal{D}_{F+G}})\leq \mathrm{Card}(\overleftrightarrow{\mathcal{D}_F})+ \mathrm{Card}(\overleftrightarrow{\mathcal{D}_G})$,
    \item $\mathrm{Card}(\overleftrightarrow{\mathcal{D}_{F\cdot G}})\leq \mathrm{Card}(\overleftrightarrow{\mathcal{D}_F})+ \mathrm{Card}(\overleftrightarrow{\mathcal{D}_G}) +n_1n_2$,
    \item $\mathrm{Card}(\overleftrightarrow{\mathcal{D}_{F^*}})\leq\mathrm{Card}(\overleftrightarrow{\mathcal{D}_F})+ 2n_1(n_1+1)$.
  \end{enumerate}
  
  On the one hand the cardinality of $\overleftrightarrow{\mathcal{D}_{F^*}}$ is strictly greater than the cardinality of $\overleftrightarrow{\mathcal{D}_{F}}$ although $F$ and $F^*$ have the same width $n_1$; we therefore substitute the parameter $m_1=n_1+h_1$ to $n_1$, so that $F^*$ is associated with $m_1+1$.
  
  On the other hand, the maximal increase of the cardinality of $\overleftrightarrow{\mathcal{D}_{E}}$ (w.r.t. $m$) occurs in the star case; we therefore consider the function $\phi$ such that:
  \begin{enumerate}
    \item $\phi(0)=0$ and $\phi(1)=1$,
    \item $\phi(k+1)=\phi(k)+2\times k\times (k+1)$,
  \end{enumerate}
and we show that $\overleftrightarrow{\mathcal{D}_{E}}\leq \phi(m)$ for any regular expression $E$.
  
  According to Lemma~\ref{lem form ind calcul deriv 2side} and by induction hypothesis, it holds:
  \begin{enumerate}
    \item $\mathrm{Card}(\overleftrightarrow{\mathcal{D}_{F+G}})\leq \phi(m_1)+\phi(m_2)$,
    \item $\mathrm{Card}(\overleftrightarrow{\mathcal{D}_{F\cdot G}})\leq \phi(m_1)+\phi(m_2)+n_1\times n_2$,
    \item $\mathrm{Card}(\overleftrightarrow{\mathcal{D}_{F^*}})\leq \phi(m_1)+2n_1(n_1+1)$.
  \end{enumerate}
  
  It can be checked that:
  
  \centerline{$\phi(m_1)+\phi(m_2)\leq \phi(m_1)+\phi(m_2)+n_1\times n_2 \leq \phi(m_1+m_2)$.}
  
  As a consequence, it holds:
  \begin{enumerate}
    \item $\mathrm{Card}(\overleftrightarrow{\mathcal{D}_{F+G}})\leq \phi(m_1+m_2)$,
    \item $\mathrm{Card}(\overleftrightarrow{\mathcal{D}_{F\cdot G}})\leq \phi(m_1+m_2)$.
  \end{enumerate} 
  
  Furthermore, by definition of $\phi$ and since $m_1\geq n_1$, it holds:  
  
  \centerline{$\phi(m_1)+2n_1(n_1+1) \leq \phi(m_1)+2(m_1)(m_1+1)=\phi(m_1+1)$}
  
  and consequently $\mathrm{Card}(\overleftrightarrow{\mathcal{D}_{F^*}})\leq \phi(m_1+1)$.

 Finally, since $\sum^k_{j=1} j(j+1)=\frac{k(k+1)(k+2)}{3}$, it holds for all integer $k\geq 1$:
 
 \centerline{
   $\phi(k)=\frac{2k(k+1)(k+2)}{3}-3$.
 }

  \cqfd
\end{proof}

\begin{proposition}\label{prop nombre der hairpin}
  Let $E$ be a regular expression over an alphabet $\Gamma$, $\mathrm{H}$ be an antimorphism over $\Gamma^*$ and $k>0$ be an integer. Then:
  \begin{enumerate}
    \item $\mathrm{Card}(\overleftrightarrow{\mathcal{D}_{\mathrm{H}'_k(E)}})\leq k\times \mathrm{Card}(\overleftrightarrow{\mathcal{D}_E})$,
    \item $\mathrm{Card}(\overleftrightarrow{\mathcal{D}_{\overrightarrow{\mathrm{H}_k}(E)}})\leq \mathrm{Card}(\overleftarrow{\mathcal{D}_E})+ k\times \mathrm{Card}(\overleftrightarrow{\mathcal{D}_E})$,
    \item $\mathrm{Card}(\overleftrightarrow{\mathcal{D}_{\overleftarrow{\mathrm{H}_k}(E)}})\leq \mathrm{Card}(\overrightarrow{\mathcal{D}_E})+ k\times \mathrm{Card}(\overleftrightarrow{\mathcal{D}_E})$.
  \end{enumerate}
\end{proposition}
\begin{proof}
  Let $E$ be a regular expression.
  
    \textbf{(1)} Let us set $\mathcal{E}=\{\mathrm{H}'_{k'}(E')\mid E'\in \overleftrightarrow{\mathcal{D}_E}\ \wedge\ k'< k \} \cup \overleftrightarrow{\mathcal{D}_E}$. Let us show that $\overleftrightarrow{\mathcal{D}_{\mathrm{H}'_k(E)}}\subset\mathcal{E}$.
    
    \textbf{(a)} According to Definition~\ref{def form deriv part}, for any couple $(x,y)$ in $\Sigma_\Gamma$, $\frac{\partial}{\partial(x,y)}(\mathrm{H}'_k(E))\subset\mathcal{E}$.
    
      \textbf{(b)} Let us show that any derived term of an expression $G$ in $\mathcal{E}$ belongs to $\mathcal{E}$.
      
      \textbf{(i)} if $G$ belongs to $ \overleftrightarrow{\mathcal{D}_E}$, so do its derived terms.
      
        \textbf{(ii)} if $G\in \{\mathrm{H}'_{k'}(E')\mid E'\in \overleftrightarrow{\mathcal{D}_E}\ \wedge\ k'< k \}$, then $G= \mathrm{H}'_{k'}(G_1)$ with $G_1\in \overleftrightarrow{\mathcal{D}_E}$ and from Definition~\ref{def form deriv part} it holds:
        
        \centerline{ $\frac{\partial}{\partial(x,y)}(G)\subset\{\mathrm{H}'_{k''}(G_2)\mid G_2\in \overleftrightarrow{\mathcal{D}_{G_1}}\ \wedge\ k''< k' \} \cup \overleftrightarrow{\mathcal{D}_{G_1}}$.}
        
         By definition of $G_1$, $\overleftrightarrow{\mathcal{D}_{G_1}}\subset \overleftrightarrow{\mathcal{D}_{E}}$. Consequently $\frac{\partial}{\partial(x,y)}(G)\subset \mathcal{E}$.
        
      \textbf{(c)} Finally, since $\mathrm{Card}(\mathcal{E})=(k-1)\times \mathrm{Card}(\overleftrightarrow{\mathcal{D}_E}) +\mathrm{Card}(\overleftrightarrow{\mathcal{D}_E})$, the proposition holds.
      
    \textbf{(2)} Let us set $\mathcal{E}=\{\overrightarrow{\mathrm{H}_k}(E')\mid E'\in \overleftarrow{\mathcal{D}_E} \} \cup \{\mathrm{H}'_{k'}(E')\mid E'\in \overleftrightarrow{\mathcal{D}_E}\ \wedge\ k'< k \} \cup \overleftrightarrow{\mathcal{D}_E}$. Let us show that $\overleftrightarrow{\mathcal{D}_{\overrightarrow{\mathrm{H}_k}(E)}}\subset\mathcal{E}$.
    
    \textbf{(a)}  According to Definition~\ref{def form deriv part}, for any couple $(x,y)$ in $\Sigma_\Gamma$, $\frac{\partial}{\partial(x,y)}(\overrightarrow{\mathrm{H}_k}(E))\subset\mathcal{E}$.
    
      \textbf{(b)} Let us show that any derived term of an expression $G$ in $\mathcal{E}$ belongs to $\mathcal{E}$.
      
      \textbf{(i)} if $G$ belongs to $\{\overrightarrow{\mathrm{H}_k}(E')\mid E'\in \overleftarrow{\mathcal{D}_E}\}$ then $G=\overrightarrow{\mathrm{H}_k}(G_1)$ with $G_1\in \overleftrightarrow{\mathcal{D}_E}$ and from Definition~\ref{def form deriv part} it holds that:
      
      \centerline{ $\frac{\partial}{\partial(x,y)}(G)\subset \{\overrightarrow{\mathrm{H}_{k}}(G_2)\mid G_2\in \overleftarrow{\mathcal{D}_{G_1}} \} \cup \{\mathrm{H}'_{k'}({G_2})\mid {G_2}\in \overleftrightarrow{\mathcal{D}_{G_1}}\ \wedge\ k'< k \} \cup \overleftrightarrow{\mathcal{D}_{G_1}}$.}
      
       Since by definition of $G_1$, $\overleftrightarrow{\mathcal{D}_{G_1}}\subset \overleftrightarrow{\mathcal{D}_{E}}$ and $\overleftarrow{\mathcal{D}_{G_1}}\subset \overleftarrow{\mathcal{D}_{E}}$, it holds: $\frac{\partial}{\partial(x,y)}(G)\subset \mathcal{E}$.
      
        \textbf{(ii)} if $G$ belongs to $\{\mathrm{H}'_{k'}(E')\mid E'\in \overleftrightarrow{\mathcal{D}_E}\ \wedge\ k'< k \}$,then $G= \mathrm{H}'_{k'}(G_1)$ with $G_1\in \overleftrightarrow{\mathcal{D}_E}$ and from Definition~\ref{def form deriv part} it holds:
        
        \centerline{ $\frac{\partial}{\partial(x,y)}(G)\subset\{\mathrm{H}'_{k''}(G_2)\mid G_2\in \overleftrightarrow{\mathcal{D}_{G_1}}\ \wedge\ k''< k' \} \cup \overleftrightarrow{\mathcal{D}_{G_1}}$.}
        
         By definition of $G_1$, $\overleftrightarrow{\mathcal{D}_{G_1}}\subset \overleftrightarrow{\mathcal{D}_{E}}$. Hence $\frac{\partial}{\partial(x,y)}(G)\subset \mathcal{E}$.
        
        \textbf{(iii)} if $G$ belongs to $\overleftrightarrow{\mathcal{D}_E}$, so do its derived terms.
        
      \textbf{(c)} Finally, since $\mathrm{Card}(\mathcal{E})=\mathrm{Card}(\overleftarrow{\mathcal{D}_E})+(k-1)\times \mathrm{Card}(\overleftrightarrow{\mathcal{D}_E}) +\mathrm{Card}(\overleftrightarrow{\mathcal{D}_E})$, the proposition holds.
    
    \textbf{(3)}The proof is similar as for case \textbf{(2)}, with $\overrightarrow{\mathcal{D}_E}$ playing the role of $\overleftarrow{\mathcal{D}_E}$.

  \cqfd
\end{proof}

  The \emph{index} of a hairpin expression $E$ is the integer $\mathrm{index}(E)$ inductively defined by: 
  
  \centerline{$\mathrm{index}(F)=0$,}
  
  \centerline{ $\mathrm{index}(\overleftarrow{\mathrm{H}_k}(F))=k$, $\mathrm{index}(\overrightarrow{\mathrm{H}_k}(F))=k$, $\mathrm{index}(\mathrm{H}'_k(F))=k$,}
  
  \centerline{ $\mathrm{index}(G_1+G_2)=\mathrm{max}(\mathrm{index}(G_1),\mathrm{index}(G_2))$,}
  
   where $\mathrm{H}$ is any anti-morphism over $\Gamma^*$, $k>0$ is any integer, $F$ is any regular expression over $\Gamma$, and $G_1$ and $G_2$ are any two hairpin expressions over $\Gamma$.

\begin{proposition}\label{prop nbre term der fini}
  Let $E$ be a hairpin expression over an alphabet $\Gamma$. Then $\overleftrightarrow{\mathcal{D}_E}$ is a finite set the cardinal of which is upper bounded by $k\times  (\frac{2m(m+1)(m+2)}{3}-3)+n$, where $k$ is the index of $E$, and $m=n+h$ with $n$ its width and $h$ its star number.
\end{proposition}
\begin{proof}
  Directly deduced from Proposition~\ref{prop nbre 2side quot reg exp} and from Proposition~\ref{prop nombre der hairpin} for the non-sum cases.  
  Whenever $E=G_1+G_2$, let us set for $i\in\{1,2\}$, $n_i$ the width of $G_i$, $h_i$ its star number, $k_i$ its index and $m_i=n_i+h_i$. Without loss of generality suppose that $k_1\geq k_2$.
  Let $\phi$ be the function defined by:
  
  \centerline{
    $\phi(k)=
      \left\{
        \begin{array}{l@{\ }l}
          0 & \text{ if }k=0,\\
         \frac{2k(k+1)(k+2)}{3}-3 & \text{ otherwise.}\\
        \end{array}
      \right.$
  }  
  
  It can be checked that the following proposition \textbf{P} holds:
  
  \centerline{
    $\phi(k_1+k_2)\geq \phi(k_1) + \phi(k_2)$.
  }  
  
  By induction and from \textbf{P} it holds:
  
  \centerline{
    \begin{tabular}{l@{\ }l}
      $\mathrm{Card}(\overleftrightarrow{\mathcal{D}_{G_1}})+\mathrm{Card}(\overleftrightarrow{\mathcal{D}_{G_2}})$ &
      $\leq k_1\times \phi(m_1)+n_1+k_2 \phi(m_2)+n_2$\\
      & $\leq k_1\times (\phi(m_1)+ \phi(m_2))+n$\\
      & $\leq k_1\times \phi(m_1+m_2)+n$\\      
    \end{tabular}
  }
  \cqfd
\end{proof}

This finite set of two-sided derived terms allows us to extend the finite derived term automaton to hairpin expressions.

\begin{definition}
  Let $E$ be a hairpin expression over an alphabet $\Gamma$. Let $A=(\Sigma_\Gamma,Q,I,F,\delta)$ be the NFA defined by: 
  \begin{itemize}
    \item $Q=\{E\}\cup \overleftrightarrow{\mathcal{D}_E}$,
    \item  $I=\{E\}$,
    \item $F=\{E'\in Q\mid \varepsilon\in L(E')\}$,
    \item  $\forall (x,y)\in\Sigma_\Gamma, \forall E' \in Q,$ $\delta(E',(x,y))= \frac{\partial}{\partial_{(x,y)}}(E')$.
  \end{itemize}
  
  The automaton $A$ is the \emph{two-sided derived term automaton} of $E$.
\end{definition}

By construction, $A$ is a $\Gamma$-couple NFA where $\Gamma$ is the alphabet of $E$.

\begin{example}
  Let $E$ be the hairpin expression of Example~\ref{ex calcul deriv term}. The derived term automaton of $E$ is the automaton presented in Figure~\ref{fig ex aut term deriv}.
\end{example}

\begin{figure}[H]
  \centerline{ 
    \begin{tikzpicture}[node distance=2cm,bend angle=30]   
    \node[initial below,state] (E) {$E$}; 
    \node[state,rounded rectangle,right of=E] (h1c) {$\overrightarrow{\mathrm{H}_1}(c)$}; 
    \node[accepting,state,left of=E] (eps) {$\varepsilon$}; 
    \node[state,rounded rectangle,right of=h1c] (h1eps) {$\overrightarrow{\mathrm{H}_1}(\varepsilon)$};
    \path[->]
      (E)   edge [swap,in=120,out=60,loop] node {$(a,a)$} ()
      (E)   edge [swap] node {$(b,c)$} (h1c)
      (E)   edge  node {$(b,c)$} (eps)
      (h1c)   edge [swap] node {$(c,b)$} (h1eps)
      ;	            
    \end{tikzpicture}
  }  
  \caption{The Derived Term Automaton of the Expression $E$.}
  \label{fig ex aut term deriv}
\end{figure}
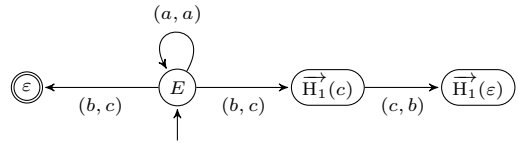

\begin{proposition}\label{prop twosided der term aut bon lang}
  Let $E$ be a hairpin expression over an alphabet $\Gamma$ and $A$ be the two-sided derived term automaton of $E$. Then: $L(E)=L_\Gamma(A)$.
\end{proposition}
\begin{proof}
  Let $A=(\Sigma,Q,I,F,\delta)$, let $w$ be a word in $\Gamma^*$ and let $E'$ be a state in $Q$. Let us show that the following proposition \textbf{(P)} is satisfied: $w\in \overrightarrow{L}_\Gamma(E')$ $\Leftrightarrow$ $w\in L(E')$.  
  By recurrence over the length of $w$.
  
  \textbf{(I)} If $w=\varepsilon$, then:
  
  $w \in \overrightarrow{L}_\Gamma(E')$
  
  $\Leftrightarrow$ $E'\in F$ (Lemma~\ref{lem calcul lang droit gamma})
  
  $\Leftrightarrow$ $\varepsilon\in L(E')$ (Construction of $A$)
  
  $\Leftrightarrow$ $w\in L(E')$.
  
  \textbf{(II)} Let us suppose that $|w|>0$. Then $\exists (x,y)\in\Sigma_\Gamma$, $\exists w'\in\Gamma^*$ such that $w=xw'y$. 
  
  \textbf{(a)} If $E'$ is a simple hairpin expression, then
  
  $xw'y \in \overrightarrow{L}_\Gamma(E')$
  
   $\Leftrightarrow$  
      $w' \in (x,y)^{-1}(\overrightarrow{L}_\Gamma(E')) $ (Definition~\ref{def two side quot})
  
  $\Leftrightarrow$ 
      $w' \in \bigcup_{E''\in\delta(E',(x,y))} \overrightarrow{L}_\Gamma(E'') $
  (Corollary~\ref{cor prop triv couple nfa})
  
  $\Leftrightarrow$ $w' \in \bigcup_{E''\in\frac{\partial}{\partial(x,y)}(E')} \overrightarrow{L}_\Gamma(E'')$ (Construction of $A$)
  
  $\Leftrightarrow$ $w' \in \bigcup_{E''\in\frac{\partial}{\partial(x,y)}(E')} L(E'')$ (Recurrence hypothesis)
  
  $\Leftrightarrow$ $w' \in (x,y)^{-1}(L(E'))$ (Proposition~\ref{prop deriv two side bon lang})
  
  $\Leftrightarrow$ $xw'y \in L(E')$ (Definition~\ref{def two side quot})  
  $\Leftrightarrow$ $w \in L(E')$.  
  
  \textbf{(b)} If $E'\in\{\overleftarrow{\mathrm{H}}_k(F),\overrightarrow{\mathrm{H}}_k(F),\mathrm{H}'_k(F)\}$, then it holds $w\in L(E')$ $\Rightarrow$ $y=\mathrm{H}(x)$ (according to Lemma~\ref{lem mot ds lang h couple deb fin}). Consequently, if $y\neq\mathrm{H}(x)$,$\delta(E',(x,y))=\emptyset$ and $w\notin \overrightarrow{L}_\Gamma(E')$. Hence, since $w\notin L(E')$, proposition is satisfied. Let us now suppose that $y=\mathrm{H}(x)$. Since $(\varepsilon,\varepsilon)\notin \Sigma_\Gamma$, $(x,y)\in \Gamma\times\Gamma$.
  
  $xw'\mathrm{H}(x) \in \overrightarrow{L}_\Gamma(E')$
  
   $\Leftrightarrow$ $w' \in (x,\mathrm{H}(x))^{-1}(\overrightarrow{L}_\Gamma(E'))$ (Definition~\ref{def two side quot})
  
  $\Leftrightarrow$ $w' \in \bigcup_{E''\in\delta(E',(x,\mathrm{H}(x)))} \overrightarrow{L}_\Gamma(E'')$ (Corollary~\ref{cor prop triv couple nfa})
  
  $\Leftrightarrow$ $w' \in \bigcup_{E''\in\frac{\partial}{\partial(x,\mathrm{H}(x))}(E')} \overrightarrow{L}_\Gamma(E'')$ (Construction of $A$)
  
  $\Leftrightarrow$ $w' \in \bigcup_{E''\in\frac{\partial}{\partial(x,\mathrm{H}(x))}(E')} L(E'')$ (Recurrence hypothesis)
  
  $\Leftrightarrow$ $w' \in (x,\mathrm{H}(x))^{-1}(L(E''))$ (Proposition~\ref{prop deriv two side bon lang})
  
  $\Leftrightarrow$ $xw'\mathrm{H}(x) \in L(E')$ (Definition~\ref{def two side quot})
  
  $\Leftrightarrow$ $w \in L(E')$
  
  Finally, 
  
  $L_{\Gamma}(A)=\bigcup_{i\in I} \overrightarrow{L}_\Gamma(i)$ (Lemma~\ref{lem prop triv couple nfa})
  
  $= \overrightarrow{L}_\Gamma(E)$ (Construction of $A$)
  
  $=L(E)$ (proposition \textbf{P}).  
  \cqfd
\end{proof}

\begin{theorem}\label{thm aut bon lang taill}
  Let $A$ be the two-sided derived term automaton of a hairpin expression $E$ over an alphabet $\Gamma$ and let $k$ be the index of $E$. Then $L_\Gamma(A)=L(E)$. Furthermore $A$ has at most $k\times ( \frac{2 m\times (m+1) \times(m+2)}{3}-3)+n+1$ states where $m=n+h$, with $n$ the width of $E$ and $h$ its star number.
\end{theorem}
\begin{proof}
  Corollary of Proposition~\ref{prop twosided der term aut bon lang} and of Proposition~\ref{prop nbre term der fini}.
  \cqfd
\end{proof}

Finally, the computation of the two-sided derived term automaton provides an alternative proof of the following theorem.

\begin{theorem}\label{thm lang haitpin line ctxt free}
  The language denoted by a hairpin expression is linear context-free.
\end{theorem}
\begin{proof}
  According to Theorem~\ref{thm context free equ recocouple} and to Proposition~\ref{prop twosided der term aut bon lang}.
  \cqfd
\end{proof}

\section{The $(\mathrm{H},0)$-Completion}\label{sec:H0}

  In the literature, the case where $k=0$ is usually not considered. Nevertheless, this case is interesting since the associated derivation computation yields 
  a recognizer with a linear number of states w.r.t. the width of the expression.
  
  Let $L_1$ and $L_2$ be two languages over an alphabet $\Gamma$ and $\mathrm{H}$ be an anti-morphism over $\Gamma^*$. The $(\mathrm{H},0)$-\emph{completion} of $L_1$ and $L_2$ is the language $\mathrm{H}_0(L_1,L_2)=\{\alpha\gamma\mathrm{H}(\alpha)\mid \alpha,\gamma\in\Gamma^*\ \wedge\ (\alpha\gamma\in L_1\ \vee\ \gamma\mathrm{H}(\alpha)\in L_2)\}$. As in the general case, the $(\mathrm{H},0)$-completion can be defined as the union of two unary operators $\overleftarrow{\mathrm{H}_0}$ and $\overrightarrow{\mathrm{H}_0}$.
  
  The \emph{left} (resp. \emph{right}) $(\mathrm{H},0)$-\emph{completion} of a language $L$ over an alphabet $\Gamma$ is the language $\overleftarrow{\mathrm{H}}_0(L)=\{\alpha\gamma\mathrm{H}(\alpha)\mid \alpha,\gamma\in\Gamma^*\ \wedge\ \gamma\mathrm{H}(\alpha)\in L\}$ (resp. $\overrightarrow{\mathrm{H}}_0(L)=\{\alpha\gamma\mathrm{H}(\alpha)\mid \alpha,\gamma\in\Gamma^*\ \wedge\ \alpha\gamma\in L\}$).
  
  Let $E$ be a regular expression over $\Gamma$ and $\mathrm{H}$ be an anti-morphism over $\Gamma^*$. The \emph{left} (resp. \emph{right}) $(\mathrm{H},0)$-\emph{completion} of $E$ is the expression $\overleftarrow{\mathrm{H}}_0(E)$ (resp.  $\overrightarrow{\mathrm{H}}_0(E)$) that denotes $\overleftarrow{\mathrm{H}}_0(L(E))$ (resp. $\overrightarrow{\mathrm{H}}_0(L(E))$).

\begin{lemma}\label{lem eps h0}
  Let $\Gamma$ be an alphabet and $\mathrm{H}$ be an anti-morphism over $\Gamma^*$. Let $L$ be a language over $\Gamma$. Then the two following conditions are satisfied:
  \begin{itemize}
    \item $\varepsilon\in \overrightarrow{\mathrm{H}}_0(L) \Leftrightarrow \varepsilon\in L$,
    \item $\varepsilon\in \overleftarrow{\mathrm{H}}_0(L) \Leftrightarrow \varepsilon\in L$.
  \end{itemize}
\end{lemma}
\begin{proof}
  Trivially proved from the definitions of left and right $(\mathrm{H},0)$-completions.
  \cqfd
\end{proof}

We now consider the construction of a recognizer for the $(\mathrm{H},0)$-completion of a regular expression $E$.
On the opposite of the general case, it is not necessary to consider the whole computation of partial derivatives. We show that it is sufficient to consider one-sided partial derivatives of regular expression.

\begin{definition}
  Let $\Gamma$ be an alphabet and $\mathrm{H}$ be an anti-morphism over $\Gamma^*$. Let $F$ be a regular expression over $\Gamma$. Let $E= \overrightarrow{\mathrm{H}}_0(F)$ (resp. $E= \overleftarrow{\mathrm{H}}_0(F)$). The \emph{effective subset associated with} $E$ is the set defined by:
  
  \centerline{
    $ \mathcal{E}=\overrightarrow{\mathrm{H}}_0(\overleftarrow{\mathcal{D}_F})\cup \overleftarrow{\mathcal{D}_F}$,
  }
  
  \centerline{
    (resp. $ \mathcal{E}=\overleftarrow{\mathrm{H}}_0(\overrightarrow{\mathcal{D}_F})\cup \overrightarrow{\mathcal{D}_F}$).
  }
\end{definition}

\begin{definition}
  Let $\Gamma$ be an alphabet and $\mathrm{H}$ be an anti-morphism over $\Gamma^*$. Let $F$ be a regular expression over $\Gamma$. Let $E= \overrightarrow{\mathrm{H}}_0(F)$ (resp. $E= \overleftarrow{\mathrm{H}}_0(F)$). Let $\mathcal{E}$ be the effective subset associated with $E$. Let $A=(\Sigma_\Gamma,Q,I,F,\delta)$ be the couple NFA defined by: $Q=\{E\}\cup \mathcal{E}$, $I=\{E\}$, $F=\{E'\in Q\mid \varepsilon\in L(E')\}$, $\forall (x,y) \in\Sigma_\Gamma, \forall E' \in Q,$

  $\delta(E',(x,y))=
    \left\{
      \begin{array}{l@{\ }l}
        \overrightarrow{\mathrm{H}}_0(\frac{\partial}{\partial_{x}}(E'')) & \text{ if } y=\mathrm{H}(x)\ \wedge\ E'=\overrightarrow{\mathrm{H}}_0(E''),\\
        \frac{\partial}{\partial_{x}}(E'') & \text{ if } y=\varepsilon\ \wedge\ E'=\overrightarrow{\mathrm{H}}_0(E''),\\
        \frac{\partial}{\partial_{x}}(E') & \text{ if } y=\varepsilon\ \wedge\ E'\text{ is a regular expression,}\\
        \emptyset & \text{ otherwise,}\\
      \end{array}
    \right.
  $

  resp. $\delta(E',(x,y))=
    \left\{
      \begin{array}{l@{\ }l}
        \overleftarrow{\mathrm{H}}_0((E'')\frac{\partial}{\partial_{y}}) & \text{ if } y=\mathrm{H}(x)\ \wedge\ E'=\overleftarrow{\mathrm{H}}_0(E''),\\
        (E'')\frac{\partial}{\partial_{y}} & \text{ if } x=\varepsilon\ \wedge\ E'=\overleftarrow{\mathrm{H}}_0(E''),\\
        (E')\frac{\partial}{\partial_{y}} & \text{ if } x=\varepsilon\ \wedge\ E'\text{ is a regular expression,}\\
        \emptyset & \text{ otherwise.}\\
      \end{array}
    \right.
  $
  
  The automaton $A$ is said to be the \emph{effective automaton} of $E$.
\end{definition}

\begin{theorem}
  Let $F$ be a regular expression over an alphabet $\Gamma$. Let $A$ be the effective automaton of the expression $E= \overrightarrow{\mathrm{H}}_0(F)$ (resp. $E= \overleftarrow{\mathrm{H}}_0(F)$). Then $L_\Gamma(A)=L(E)$. Furthermore $A$ has at most $2n+1$ states where $n$ is the width of $E$.
\end{theorem}
\begin{proof}
  Let us set $A=(\Sigma_\Gamma,Q,I,F,\delta)$.
  
  \textbf{(I)} Let us show now that $L_\Gamma(A)=L(E)$. 
  
  \textbf{(a)} Let us suppose that $E= \overrightarrow{\mathrm{H}}_0(F)$. Let $w$ be a word in $\Gamma^*$. Let us show by recurrence over the length of $w$ that for any state $E'$ in $Q$,  $w\in L(E') \Leftrightarrow$ $w\in \overrightarrow{L}_\Gamma(E')$.
  
  \textbf{(1)} If $w=\varepsilon$, $w\in L(E') $ $\Leftrightarrow$ $E'\in F$ $\Leftrightarrow$ $w\in \overrightarrow{L}_\Gamma(E')$.
  
  \textbf{(2)} Let $w$ be a word different from $\varepsilon$. 
  
  \textbf{(i)} If $E'$ is a regular expression,  $a^{-1}(L(E'))= \bigcup_{ E''\in\frac{\partial}{\partial_a}(E')} L(E'')$. Hence since there exists $a$ in $\Gamma$ and $w'$ in $\Gamma^*$ such that $w=aw'$, it holds:
  
   $aw'\in L(E')$ $\Leftrightarrow$ $w'\in a^{-1}(L(E'))$ $\Leftrightarrow$ $w'\in \bigcup_{ E''\in\frac{\partial}{\partial_a}(E')} L(E'')$
   
   $\Leftrightarrow$ $w'\in \bigcup_{ E''\in\frac{\partial}{\partial_a}(E')} \overrightarrow{L}_\Gamma(E'')$
   
   $\Leftrightarrow$ $w'\in \bigcup_{ E''\in\delta(E',(a,\varepsilon))} \overrightarrow{L}_\Gamma(E'')$ (Recurrence Hypothesis)
   
   $\Leftrightarrow$ $aw'\in \overrightarrow{L}_\Gamma(E')$.
  
   \textbf{(ii)} If $E'=\overleftarrow{\mathrm{H}}_0(E'')$ then:
  
   $w\in L(\overleftarrow{\mathrm{H}}_0(E''))$ $\Leftrightarrow$ $\exists \alpha,\gamma\in \Gamma^*, (w=\alpha\gamma\mathrm{H}(\alpha)\ \wedge\ \alpha\gamma\in L(E''))$ 
   
   $\Leftrightarrow$ $\exists a\in\Gamma,\gamma\in \Gamma^*,\alpha'\in \Gamma^*, ((w=\gamma\ \wedge\ \gamma\in L(E''))\ \vee\ (w=a\alpha'\gamma\mathrm{H}(\alpha')\mathrm{H}(a)\ \wedge\ a\alpha'\gamma\in L(E''))$
   
   $\Leftrightarrow$ $\exists a\in\Gamma,\gamma\in \Gamma^*,\alpha'\in \Gamma^*,w'\in \Gamma^*, ((w=aw'\ \wedge\ w'\in a^{-1}(L(E'')))\ \vee\ (w=a\alpha'\gamma\mathrm{H}(\alpha')\mathrm{H}(a)\ \wedge\ \alpha'\gamma\in a^{-1}(L(E''))))$
   
   $\Leftrightarrow$ $\exists a\in\Gamma,\gamma\in \Gamma^*,\alpha'\in \Gamma^*,w'\in \Gamma^*, ((w=aw'\ \wedge\ w'\in \bigcup_{ E''\in\frac{\partial}{\partial_a}(E')} L(E''))\ \vee\ (w=a\alpha'\gamma\mathrm{H}(\alpha')\mathrm{H}(a)\ \wedge\ \alpha'\gamma\in \bigcup_{ E''\in\frac{\partial}{\partial_a}(E')} L(E'')))$ 
   
   $\Leftrightarrow$ $\exists a\in\Gamma,\gamma\in \Gamma^*,\alpha'\in \Gamma^*,w'\in \Gamma^*, ((w=aw'\ \wedge\ w'\in \bigcup_{ E''\in\frac{\partial}{\partial_a}(E')} \overrightarrow{L}_\Gamma(E''))\ \vee\ (w=a\alpha'\gamma\mathrm{H}(\alpha')\mathrm{H}(a)\ \wedge\ \alpha'\gamma\in \bigcup_{ E''\in\frac{\partial}{\partial_a}(E')} \overrightarrow{L}_\Gamma(E'')))$ (Recurrence Hypothesis)
   
   $\Leftrightarrow$ $\exists a\in\Gamma,\gamma\in \Gamma^*,\alpha'\in \Gamma^*,w'\in \Gamma^*, ((w=aw'\ \wedge\ w'\in \bigcup_{ E''\in\delta(E',(a,\varepsilon))} \overrightarrow{L}_\Gamma(E''))\ \vee\ (w=a\alpha'\gamma\mathrm{H}(\alpha')\mathrm{H}(a)\ \wedge\ \alpha'\gamma\in \bigcup_{ E''\in\delta(E',(a,\mathrm{H}(a)))} \overrightarrow{L}_\Gamma(E'')))$  
   
   $\Leftrightarrow$ $w\in \overrightarrow{L}_\Gamma(E')$.
   
   Finally since $L(A)= \overrightarrow{L}_\Gamma(E)$ and since $L(E)= \overrightarrow{L}_\Gamma(E)$, then $L(A)=L(E)$.
   
   \textbf{(b)} The case where $E= \overleftarrow{\mathrm{H}}_0(F)$ is based on the same reasoning.   

  \textbf{(II)} Let $\mathcal{E} =\overrightarrow{\mathrm{H}}_0(\overleftarrow{\mathcal{D}_F})\cup \overleftarrow{\mathcal{D}_F}$ be the effective subset associated with $E$ (resp. $\mathcal{E} =\overleftarrow{\mathrm{H}}_0(\overrightarrow{\mathcal{D}_F})\cup \overrightarrow{\mathcal{D}_F}$). Since $\mathrm{Card}(\overleftarrow{\mathcal{D}_F})\leq n$ (resp. $\mathrm{Card}(\overrightarrow{\mathcal{D}_F})\leq n$), the number of states of $A$ is at most $2n$. Finally, since $Q=\mathcal{E}\cup\{E\}$, it holds that $A$ has at most $2n+1$ states.
  \cqfd
\end{proof}

\begin{example}
  Let $\mathrm{H}$ be the anti-morphism defined in Example~\ref{ex calcul deriv term}. Let $E=\overrightarrow{\mathrm{H}_0}(a^*bc)$. Notice that $\overleftarrow{\mathcal{D}}_{a^*bc}=\{a^*bc,c,\varepsilon\}$. Hence the effective subset associated with $E$ is the set $\{\overrightarrow{\mathrm{H}_0}(a^*bc),\overrightarrow{\mathrm{H}_0}(c),\overrightarrow{\mathrm{H}_0}(\varepsilon),a^*bc,c,\varepsilon\}$.
  
   The effective automaton $A$ of $E$ is given Figure~\ref{fig ex h0 aut eff}.
  
   It can be checked that $L(A)=\{a^nbc\mid n\in\mathbb{N}\}\cup\{a^nbca^n\mid n\in\mathbb{N}\}\cup\{a^nbcca^n\mid n\in\mathbb{N}\}\cup\{a^nbcbca^n\mid n\in\mathbb{N}\}$ that is exactly $L(E)$ (see Table~\ref{table ex alpha beta}).
\end{example}

\begin{table}[H]
  \centerline{
    \begin{tabular}{c@{\ }|@{\ }c@{\ }|@{\ }c}
      $\alpha$ & $\gamma$ & $\mathrm{H}(\alpha)$\\
      \hline
      $\varepsilon$ & $a^nbc$ & $\varepsilon$\\
      $a^n$ & $bc$ & $a^n$\\
      $a^nb$ & $c$ & $ca^n$\\
      $a^nbc$ & $\varepsilon$ & $bca^n$\\
    \end{tabular}   
  }
    \caption{The Language $L(E)$}
    \label{table ex alpha beta} 
\end{table}

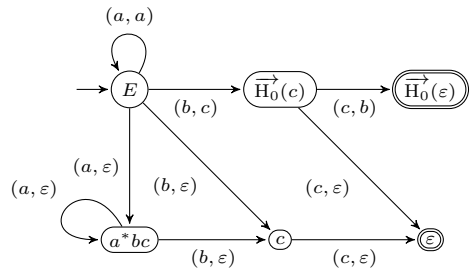
\begin{figure}[H]
  \centerline{ 
    \begin{tikzpicture}[node distance=2cm,bend angle=30]   
    \node[initial,state] (E) {$E$}; 
    \node[state,rounded rectangle,right of=E] (h0c) {$\overrightarrow{\mathrm{H}_0}(c)$}; 
    \node[state,accepting,rounded rectangle,right of=h0c] (h0eps) {$\overrightarrow{\mathrm{H}_0}(\varepsilon)$}; 
    \node[state,rounded rectangle,below of=E] (abc) {$a^*bc$};
    \node[state,rounded rectangle,right of=abc] (c) {$c$};
    \node[state,accepting,rounded rectangle, right of=c] (eps) {$\varepsilon$};    
    \path[->]
      (E)   edge [swap,in=120,out=60,loop] node {$(a,a)$} ()
      (abc)   edge [swap,in=180,out=120,loop] node {$(a,\varepsilon)$} ()
      (E)   edge [swap] node {$(b,c)$} (h0c)
      (h1c)   edge [swap] node {$(c,b)$} (h0eps)
      (E)   edge [swap] node {$(a,\varepsilon)$} (abc)
      (E)   edge [swap] node {$(b,\varepsilon)$} (c)
      (abc)   edge [swap] node {$(b,\varepsilon)$} (c)
      (c)   edge [swap] node {$(c,\varepsilon)$} (eps)
      (h0c)   edge [swap] node {$(c,\varepsilon)$} (eps)
      ;	            
    \end{tikzpicture}
  }  
  \caption{The Effective Automaton of the Expression $E$}
  \label{fig ex h0 aut eff}
\end{figure}

\section{Conclusion}

This paper provides an alternative proof of the fact that hairpin completions of regular languages are linear context-free.
This proof is obtained by considering the family of regular expressions extended to hairpin operators 
and by computing their partial derivatives, a technique that has already been applied
to regular expressions extended to boolean operators~\cite{CCM11b},
to multi-tilde-bar operators~\cite{CCM12c}
and to approximate operators~\cite{CJM12}.
Moreover it is a constructive proof since it is based on the computation of a polynomial size recognizer for hairpin completions of regular languages.
We also proved that it is possible to compute a linear size recognizer for $(H,0)$-completions of regular languages.

\bibliography{biblio}

\end{document}